\documentclass[journal,twoside,web]{ieeecolor}
\usepackage{generic}
\usepackage{cite}
\usepackage{amsmath,amssymb,amsfonts}
\usepackage{algorithmic}
\usepackage{graphicx}
\usepackage{textcomp}
\usepackage{bbm}
\usepackage[ruled,vlined]{algorithm2e}
\usepackage{thmtools, thm-restate}
\usepackage{mathtools}

\def\BibTeX{{\rm B\kern-.05em{\sc i\kern-.025em b}\kern-.08em
    T\kern-.1667em\lower.7ex\hbox{E}\kern-.125emX}}
\usepackage{hyperref}
\usepackage[capitalise,nameinlink]{cleveref}

\crefformat{equation}{#2(#1)#3}
\crefrangeformat{equation}{(#3#1#4)--(#5#2#6)}
\crefformat{section}{#2\S#1#3}
\crefformat{subsection}{#2\S#1#3}
\crefformat{subsubsection}{#2\S#1#3}
\crefname{figure}{Figure}{Figures}
\crefname{assumption}{Assumption}{Assumptions}

\newcommand{\rev}[1]{#1}

\newcommand{\E}{\mathbb{E}}

\newcommand{\R}{\mathbb{R}}

\newcommand{\calX}{\mathcal{X}}
\newcommand{\calV}{\mathcal{V}}
\newcommand{\calF}{\mathcal{F}}

\newcommand{\calW}{\mathcal{W}}
\newcommand{\calU}{\mathcal{U}}
\newcommand{\calD}{\mathcal{D}}

\newcommand{\what}{\widehat{W}}

\newcommand{\calFhat}{\widehat{\mathcal{F}}}
\newcommand{\calDhat}{\widehat{\mathcal{D}}}

\newcommand{\iid}{\overset{\textup{iid}}{\sim}}

\newcommand{\wtilde}{\widetilde{W}}
\newtheorem{remark}{Remark}
\newtheorem{lemma}{Lemma}
\newtheorem{theorem}{Theorem}
\newtheorem{corollary}{Corollary}
\newtheorem{definition}{Definition}
\newtheorem{assumption}{Assumption}

\DeclareMathOperator*{\argmin}{arg\,min}

\DeclareMathOperator*{\minimize}{\mathrm{minimize}}
\DeclareMathOperator*{\subjectto}{\mathrm{subject~to}}

\renewcommand{\baselinestretch}{0.98}
\begin{document} 
\title{Adaptive Robust Model Predictive Control\\ via Uncertainty Cancellation}
\author{Rohan Sinha, James Harrison, Spencer M. Richards, and Marco Pavone
\thanks{%
    The authors are with the Autonomous Systems Lab at Stanford University, Stanford, CA, \texttt{\{rhnsinha, jharrison, spenrich, pavone\}@stanford.edu}. This research was supported in part by the National Science Foundation (NSF) via Cyber-Physical Systems (CPS) award \#1931815, and the National Aeronautics and Space Administration (NASA) via University Leadership Initiative grant \#80NSSC20M0163 and via an Early Stage Innovations grant. Spencer~M.~Richards and James Harrison were also supported in part by the Natural Sciences and Engineering Research Council of Canada (NSERC). This article solely reflects our own opinions and conclusions, and not those of any NSF, NASA, or NSERC entity.
    This paper extends an earlier conference version of this work \cite{SinhaHarrisonEtAl2021}. 
}%
}
\maketitle

\begin{abstract}
We propose a learning-based robust predictive control algorithm that compensates for significant uncertainty in the dynamics for a class of discrete-time systems that are nominally linear with an additive nonlinear component. Such systems commonly model the nonlinear effects of an unknown environment on a nominal system. We optimize over a class of nonlinear feedback policies inspired by certainty equivalent "estimate-and-cancel" control laws pioneered in classical adaptive control to achieve significant performance improvements in the presence of uncertainties of large magnitude, a setting in which existing learning-based predictive control algorithms often struggle to guarantee safety. In contrast to previous work in robust adaptive MPC, our approach allows us to take advantage of structure (i.e., the numerical predictions) in the a priori unknown dynamics learned online through function approximation. Our approach also extends typical nonlinear adaptive control methods to systems with state and input constraints even when we cannot directly cancel the additive uncertain function from the dynamics. We apply contemporary statistical estimation techniques, to certify the system's safety through persistent constraint satisfaction with high probability. 
Moreover, we propose using Bayesian meta-learning algorithms that learn calibrated model priors to help satisfy the assumptions of the control design in challenging settings. Finally, we show in simulation that our method can accommodate more significant unknown dynamics terms than existing methods and that the use of Bayesian meta-learning allows us to adapt to the test environments more rapidly.
\end{abstract}


\section{Introduction}
\rev{
Developing control systems capable of autonomous operation in diverse, unstructured environments requires control algorithms that learn from experience. Therefore, rapid advances in machine learning algorithms (e.g., see \cite{rasmussen2006gaussian, goodfellow2016deep}) have driven a concomitant explosion in research on the use of learning algorithms to control dynamical systems (e.g., see \cite{RichardsAzizanEtAl2021, AswaniGonzalezEtAl2013, bujarbaruah_adaptive_2018, bujarbaruah_semi-definite_2020, BujarbaruahZhangEtAl2020, chowdhary_recursively_2010, joshi_asynchronous_2020, HewingKabzanEtAl2017, recht_review, lew_safe_2020, soloperto_learning-based_2018, cairano_indirect_2016, DeanManiaEtAl2018, DeanTuEtAl2019, fan_bayesian_2020, HarrisonSharmaEtAl2018, koller_learning-based_2018, mania2020active, MishraGasparino2021}). These \emph{learning-based} control algorithms leverage operational data to improve closed-loop performance, typically by refining estimates of the environments' nonlinear and a priori unknown effects on the dynamics online. 
}

\rev{
Upon deployment, these methods should provide rigorous safety guarantees while quickly adapting in the face of uncertainty. The problem of safe learning in control was initially considered in adaptive control theory, a mature discipline that has historically emphasized safety in the form of closed-loop stability guarantees \cite{Astrom2013, SlotineLi1991, IoannouSun2012}. These classical adaptive control methods have seen continued interest in recent work on learning to control for their simplicity and asymptotic convergence behavior for systems with matched uncertainty (e.g., see \cite{joshi_asynchronous_2020, ChowdharyKingraviEtAl2014, fan_bayesian_2020, BoffiTuEtAl2020, RichardsAzizanEtAl2021, OConnellShiEtAl2021}). However, these approaches generally cannot guarantee the satisfaction of constraints on states and inputs, even though this finer-grained notion of safety is essential in practice to ensure unsafe regions of the state space are avoided under limits on the control authority.
}
\begin{figure}
    \centering
    \includegraphics[width=\linewidth]{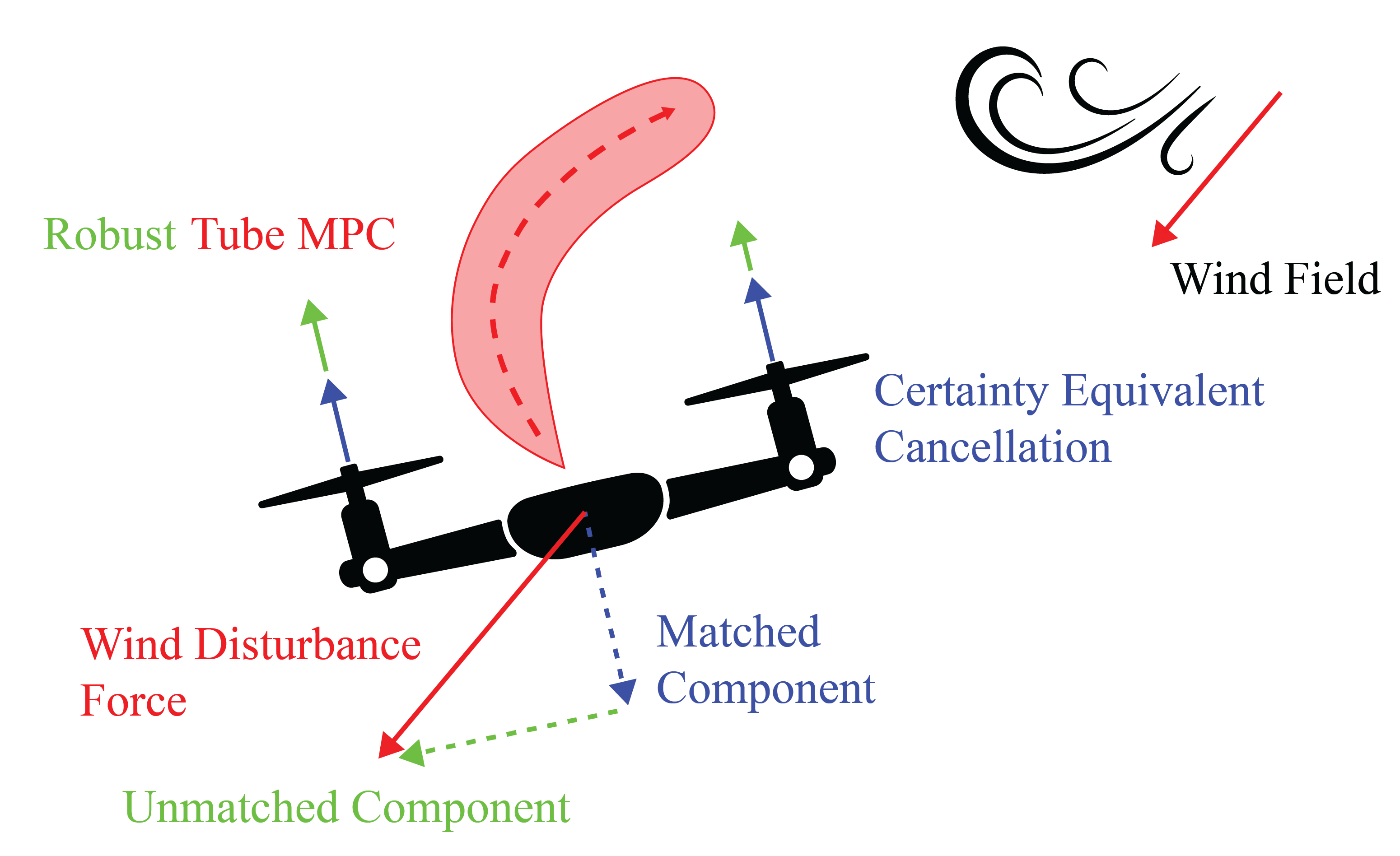}
    \caption{\rev{High-level illustration of our approach. We learn the effect of environmental uncertainty online and curb its influence through certainty equivalent adaptive control. Then, we account for any remaining uncertainty and guarantee constraint satisfaction using tube MPC.}}
    \label{fig:illustration}
\end{figure}

\rev{
In contrast, model predictive control (MPC) algorithms were developed to provide such set-avoidance guarantees \cite{BorrelliBemporadEtAl2017}. Therefore, many recently proposed learning-based control algorithms strive to integrate online learning algorithms with robust MPC strategies (e.g., see \cite{bujarbaruah_adaptive_2018, bujarbaruah_semi-definite_2020, BujarbaruahZhangEtAl2020, AswaniGonzalezEtAl2013, AswaniTomlin2012, HewingKabzanEtAl2017, lew_safe_2020, soloperto_learning-based_2018, cairano_indirect_2016, koller_learning-based_2018, KohlerAndinaEtAl2019}). However, guaranteeing constraint satisfaction then generally requires estimating uncertainty in learned quantities and propagating this uncertainty forward in time to characterize the set (or distribution) of possible trajectories when performing trajectory optimization. This is a challenging problem when we estimate the dynamics with expressive models like deep neural networks. Thus, we reach a central tension underlying modern learning-based control; we desire complex nonlinear models with the broad representational capacity necessary for autonomy in diverse and a priori unmodelled environments, but these models are not readily integrated into constrained control algorithms. This tension often results in learning-based control algorithms for constrained systems that are either too conservative (e.g., yielding limited performance to remain safe) or too fragile (e.g., infeasible in the face of considerable uncertainty).
}

\rev{
In this work, we leverage classical adaptive control techniques to reduce the over-conservatism and fragility of learning-based control algorithms for systems subject to state and input constraints.  We combine a simple nonlinear control law inspired by ``estimate-and-cancel'' methods in nonlinear adaptive control with robust MPC techniques to control a system in an uncertain environment, represented as an unknown nonlinear term in the dynamics. This strategy results in a simple control algorithm that is recursively feasible, input-to-state stable, and can safely leverage expressive nonlinear models. 
We can view our approach through the lens of both adaptive control and robust adaptive MPC. On the one hand, we extend classical adaptive cancellation-based methods to a setting with uncertain, unmatched dynamics subject to state and input constraints. On the other hand, we introduce a simple nonlinear feedback law to construct an adaptive robust MPC scheme that can reduce the conservatism of existing approaches by taking advantage of the learned structure in a priori unknown dynamics. We demonstrate on various simulated systems that our method reduces the conservatism and increases the feasible domain of the resulting robust MPC problem compared to typical adaptive robust MPC methods. 
}

\subsection{Related Work} 
\rev{We briefly review two significant paradigms for the control of uncertain systems, namely \emph{adaptive control} and \emph{robust control}. We then discuss recent works that combine ideas from both paradigms, oftentimes leveraging modern methods in machine learning.}

\subsubsection{Adaptive Control} 
\rev{Adaptive control concerns the joint design of a parametric feedback controller and a parameter adaptation law to improve closed-loop performance over time when the dynamics are partially unknown \cite{SlotineLi1991, IoannouSun2012}. Design of these components for nonlinear systems commonly relies on expressing unknown dynamics terms as linear combinations of \emph{known} basis functions, i.e., \emph{features} \cite{SlotineLi1991}. The adaptation law updates the feature weights online, and the controller applies part of the control signal to cancel the estimated term from the dynamics \cite{IoannouFidan2006, LavretskyWise2013, SlotineLi1991, Astrom2013}. These simple methods can achieve tracking convergence up to an error threshold that depends on the representation capacity of the features relative to the true dynamics~\cite{IoannouFidan2006, LavretskyWise2013}. Recent works propose combining high-capacity parametric and non-parametric models from machine learning with classical adaptive control designs. This includes deep neural networks via online back-propagation \cite{joshi_asynchronous_2020}, Gaussian processes \cite{ChowdharyKingraviEtAl2014}, and Bayesian neural networks \cite{fan_bayesian_2020, harrison_meta-learning_2020} via online Bayesian updates and meta-learned features \cite{HarrisonSharmaEtAl2018, RichardsAzizanEtAl2021}. However, these approaches are fundamentally limited by common assumptions in classical adaptive control, namely that uncertain dynamics terms can be stably canceled by the control input in their entirety, i.e., that these terms are \emph{matched uncertainties} \cite{SlotineLi1991, IoannouFidan2006, IoannouSun2012, LavretskyWise2013}. Moreover, most of these works do not consider state and input constraints, which are essential to safe control in practice. We generalize these classical adaptive methods to incorporate safety constraints even if the uncertainty is not fully matched.
}

\subsubsection{Robust Control}
\rev{Robust control seeks consistent performance despite uncertainty in the dynamics. In this work, we consider the robust control of constrained discrete-time systems using tools from predictive control. In particular, robust MPC algorithms for linear systems consider the control of a system subject to bounded noise or uncertain dynamics terms, i.e., \emph{disturbances}, as an optimization program with explicit state and input constraints. Some methods optimize the worst-case performance of the controller \cite{mayne:minmax}, while others tighten the constraints to accommodate the set of all possible trajectories induced by the disturbances and optimize the nominal predicted trajectory instead \cite{mayne_robust_2005, LimonAlamoEtAl2009}. To account for future information gain and reduce conservatism, these methods either fix a disturbance feedback policy \cite{mayne_robust_2005} or optimize over state feedback policies \cite{goulart_optimization_2006}.} 

\subsubsection{Adaptive Robust MPC (ARMPC)} 
\rev{ARMPC, often referred to as learning-based MPC, incorporates the online estimation (i.e., learning) from adaptive control methods into robust MPC to satisfy constraints in the presence of process noise and model uncertainty during learning. Recent years have seen a flurry of work on nonlinear predictive control methods that apply contemporary machine learning techniques to learn uncertain dynamics online \cite{lew_safe_2020, HewingKabzanEtAl2017, koller_learning-based_2018, MishraGasparino2021}. These methods typically result in non-convex programs for trajectory optimization under the learned dynamics, while relying on conservative approximate methods for uncertainty propagation to guarantee constraint satisfaction. However, it is unclear how to construct the necessary components, i.e., the robust positive invariant and the terminal cost function, for predictive control to make claims of persistent constraint satisfaction (i.e., safety) or stability for arbitrary nonlinear systems. Some methods ignore these topics and do not make rigorous safety guarantees \cite{HewingKabzanEtAl2017}. Other works, such as \cite{koller_learning-based_2018, MishraGasparino2021}, assume these ingredients already exist or only consider trajectory optimization tasks where a goal region needs to be reached in a finite number of time steps \cite{lew_safe_2020}. Moreover, iterative methods used to solve for local minima of non-convex programs can be computationally prohibitive and often have limited performance guarantees. }

\rev{To make rigorous safety guarantees, we will focus on adaptive robust methods for systems that are nominally linear, as considered in \cite{bujarbaruah_adaptive_2018, BujarbaruahZhangEtAl2020, bujarbaruah_semi-definite_2020, soloperto_learning-based_2018, cairano_indirect_2016, KohlerAndinaEtAl2019, AswaniGonzalezEtAl2013}. A straightforward approach is to maintain an outer bound on any unknown nonlinear terms in the dynamics and use it as a disturbance bound in any chosen robust MPC scheme \cite{AswaniGonzalezEtAl2013, soloperto_learning-based_2018,  bujarbaruah_semi-definite_2020, bujarbaruah_adaptive_2018, BujarbaruahZhangEtAl2020}. These methods avoid some of the difficulties associated with trajectory optimization for nonlinear dynamics by ignoring the actual values of the nonlinear terms at any point in the state space. That is, these methods do not exploit the learned structure in the a priori unknown dynamics, often rendering them overly conservative or fragile.}

\subsection{Contributions}
\rev{We present an ARMPC method for systems with an additive unknown nonlinear dynamics term, subject to state and input constraints. Rather than construct an outer envelope for such terms, as is normative in ARMPC literature for linear systems \cite{AswaniTomlin2012,AswaniGonzalezEtAl2013}, we develop theoretical guarantees for a broad class of \emph{function approximators}, including set membership and least-squares methods for certain noise models. Our key idea is to decompose uncertain dynamics terms into a \emph{matched} component that lies in a subspace that can be stably canceled by the control input, and an \emph{unmatched} component that lies in an orthogonal complement to this subspace. We apply certainty equivalent adaptive control techniques to stably cancel the \emph{matched} component from the dynamics and then apply robust MPC, considering the unmatched component as a bounded disturbance. Therefore, our method explicitly uses estimates of the unknown dynamics term throughout the state space for control, i.e., it takes advantage of the learned structure in the dynamics. We prove our method is recursively feasible and input-to-state stable. Moreover, we demonstrate on various simulated systems that our method reduces the conservatism and increases the feasible domain of the resulting robust MPC problem compared to typical adaptive robust MPC methods.} 

\rev{The performance of the adaptive control strategy, which learns a structured representation of the unknown dynamics term, relies on the quality of the features used in online learning. Thus, in addition to investigating standard techniques for feature construction, we introduce a Bayesian meta-learning algorithm \cite{harrison_meta-learning_2020} for feature learning. This method, which learns features that are broadly useful across tasks (or instantiations of unknown dynamics) produces useful features and well-calibrated priors, and satisfies our desiderata for learning algorithms. Beyond the central algorithmic contributions of this paper, we also show the utility of these meta-learning algorithms in adaptive and learning-based control.}

\subsection{Organization}
In \cref{sec:prob}, we pose the general robust infinite-horizon optimal control problem central to this work. In \cref{sec:prelim}, we discuss a standard robust MPC solution to this problem that is a core tool in the development of our approach. Moreover, we outline several tools that are used in the analysis of our approach. In \cref{sec:approach}, we describe our approach, prove the stability of the method, and discuss different possible assumptions on the learning setting and the impact on the controller. In \cref{sec:est}, we discuss two estimators that we use in our experiments -- set-membership and linear regression. We also compare these estimators, and introduces meta-learning as a powerful black-box tool for feature learning and prior calibration. 
Finally, we present simulation results in \cref{sec:exp}, and we conclude the paper and discuss directions for future work in \cref{sec:conc}.

\section{Problem Formulation}\label{sec:prob}
We consider the robust control of nonlinear discrete-time systems of the form
\begin{equation}\label{eq:problem-dynamics}
    x(t+1) = Ax(t) + Bu(t) + f(x(t)) + v(t),
\end{equation}
where~$x(t) \in \R^n$ is the system state, $u(t) \in \R^m$ is the control input, $A \in \R^{n \times n}$ and~$B \in \R^{n\times m}$ are known constant matrices, and $v(t) \in \calV$ is a disturbance in a known compact convex set~$\mathcal{V}$ containing the origin. In addition, an \emph{unknown, nonlinear} dynamics term $f : \R^n \to \R^n$ acts on the system, representing the unmodelled influence of the environment on the nominally linear dynamics of system~\cref{eq:problem-dynamics}. For example, $f(x)$ can model the effect that wind conditions have on the linearized dynamics of a drone. We assume the disturbances have zero mean and are independent and identically distributed (iid) according to some distribution $p(v)$, i.e., $v(t) \iid p(v)$ and~$\E[v(t)] =0$ for all $t \geq 0$. Our goal is to regulate the system to the origin according to the robust optimal control problem
\begin{equation}\label{eq:problem}
\hspace{-0.65em}\begin{aligned}
    \minimize_{x,u}\enspace&
        \E\Big[ \sum_{t=0}^\infty h(x(t), u(t)) \Big] 
    \\
    \subjectto\enspace
        &x(t+1) = Ax(t) + Bu(t) + f(x(t)) + v(t)
        \\&u(t) \in \calU,\, x(t) \in \calX,\, v(t) \in \mathcal{V},\, \forall t\in\mathbb{N}_{\geq 0}
\end{aligned}\,,
\end{equation}
where $\calX \subseteq \R^n$ and $\calU \subseteq \R^m$ are compact convex sets containing the origin,
and $h(x,u) =  x^\top Q x + u^\top R u$ is a quadratic stage cost parameterized by positive semi-definite matrix $Q \in \mathbb{S}^n_{\succeq 0}$ and positive-definite matrix $R \in \mathbb{S}^m_{\succ 0}$. The problem \cref{eq:problem} is computationally intractable to solve because the horizon is infinite and the nonlinear function~$f$ makes the problem non-convex. To approximately solve \cref{eq:problem}, we need additional assumptions on the unknown, nonlinear dynamics term~$f$. In particular, to derive a controller that is robust to any possible value of~$f(x)$, we need~$f$ to be bounded on~$\mathcal{X}$. Moreover, to construct guarantees on the online estimation of~$f$ and establish properties of a controller using this estimate, we also need to assume some structure of~$f$. For these reasons, we make the following assumption.

\begin{assumption}[structure]\label{as:bounded-feat}
    The nonlinear dynamics term ${f : \R^n \to \R^n}$ is linearly parameterizable, i.e.,
    \begin{equation}\label{eq:func-approx}
        f(x) = W\phi(x),\ \forall x \in \R^n,
    \end{equation}
    where $\phi : \R^n \to \R^d$ is a \emph{known} nonlinear feature map, and $W \in \R^{n \times d}$ is an \emph{unknown} weight matrix. Moreover, ${\|\phi(x)\| \leq 1}$ for any  $x \in \calX$, where $\|\cdot\|$ is the Euclidean~norm.
\end{assumption}

Representing a nonlinear function using a feature map is common both in adaptive control \cite{SlotineLi1991, IoannouFidan2006} and contemporary machine learning \cite{harrison_meta-learning_2020, mania2020active}, as they can represent arbitrary functions if properly designed. Without loss of generality, we assume the upper norm bound on the features is one for simplicity. This parameterization admits function classes such as neural networks with scaled sigmoid outputs.

\subsection{Matched and Unmatched Uncertainty}
While it is common in adaptive control to assume the uncertain function~$f$ in \cref{eq:problem-dynamics} can be stably cancelled in its entirety \cite{SlotineLi1991, LavretskyWise2013}, we will generalize this approach to a setting where perfect cancellation is not possible. We use the following definition to distinguish between components of the uncertain dynamics~$f$ that can and cannot be cancelled.
\begin{definition}[matched and unmatched uncertainty] 
    The uncertain function $f(x)$ in \cref{eq:problem-dynamics} is a \emph{matched uncertainty} if $f(x) \in \mathrm{Range}(B)$ for all $x \in \calX$. Conversely, if there exists an $x \in \calX$ such that $f(x) \notin \mathrm{Range}(B)$, then $f(x)$ is an \emph{unmatched uncertainty}.
\end{definition}

In this work, we assume the matrix $B$ in \cref{eq:problem-dynamics} has full column rank, i.e., there are no redundant actuators; this guarantees that the Moore-Penrose pseudoinverse $B^\dagger \coloneqq (B^\top B)^{-1}B^\top$ exists. If $f(x)$ is a matched uncertainty, then the function $g(x) = B^\dagger f(x)$ satisfies $f(x) = Bg(x)$ for any $x \in \calX$. 

Controlling systems with matched uncertainty is a classical problem in the adaptive control literature, much of which relies on the observation that setting $u(t) = \bar{u}(t) - g(x(t))$ in~\cref{eq:problem-dynamics} would cancel the nonlinear term to yield linear dynamics with respect to the nominal input $\bar{u}(t)$. \emph{Certainty equivalent} controllers approximately cancel $g(x)$ with an estimate $\hat{g}(x)$ and can yield simple nonlinear adaptive laws that achieve asymptotic tracking performance for matched systems. Even though systems are often designed to be easy to control, unmatched uncertainty affects many practical systems of interest, such as underactuated robots (e.g., quadrotors and cars). We propose to decompose the uncertain function $f$ into matched and unmatched components, apply certainty equivalent cancellation to the matched component, and curb the impact of the unmatched component with robust MPC. Applying part of the input to cancel matched uncertainty instantaneously prevents part of $f$ from leaking into the dynamics, avoiding the need to react to large observed disturbances.

\section{Robust MPC Background}\label{sec:prelim}
We now briefly review how we could approximately solve the optimal control problem \cref{eq:problem} with existing robust MPC techniques that treat $d(t) \coloneqq f(x(t)) + v(t)$ as a single bounded disturbance term, since $d(t)$ lies in the set
\begin{equation}
    \mathcal{D} \coloneqq \{f(x) + v \in \R^n \mid x \in \mathcal{X},\, v \in \mathcal{V}\}
\end{equation}
for all~$t \in \mathbb{N}_{\geq0}$. Indeed, $\mathcal{D}$ is bounded under \cref{as:bounded-feat} and the boundedness of $\mathcal{X}$ and $\mathcal{V}$.

\subsection{Receding Horizon Control}\label{sec:receding-horizon-control}
We focus on receding-horizon robust MPC schemes, whereby an approximate version of \cref{eq:problem} with finite horizon $N \in \mathbb{N}_{>0}$ is solved online with full state feedback. Rather than search over open-loop input sequences, which can incur issues with feasibility and stability, we search over closed-loop feedback policies \cite{goulart_optimization_2006, BujarbaruahZhangEtAl2020, mayne_robust_2005}. In particular, we follow \cite{goulart_optimization_2006} in optimizing over time-varying, causal, affine disturbance feedback policies of the form 
\begin{equation}\label{eq:unmatched-causal-pol}
    u_{t+k|t} = \bar{u}_{t+k|t} + {\textstyle\sum_{j=0}^{k-1}} K_{kj|t} d_{t+j|t},
\end{equation}
via the robust MPC problem
\begin{equation}\label{eq:robust-mpc}
\hspace{-0.7em}\begin{aligned}
    \minimize_{\substack{
        \{K_{kj|t}\}_{k=0,j=0}^{N-1,k-1},\\
        \{\bar{u}_{t+k|t}\}_{k=0}^{N-1}
    }}\enspace
    &V_N(\bar{x}_{t+N|t}) + \sum_{k=0}^{N-1} h(\bar{x}_{t+k|t}, \bar{u}_{t+k|t})
    \\
    \subjectto\enspace
    & \bar{x}_{t+k+1|t} = A\bar{x}_{t+k|t} + B\bar{u}_{t+k|t}
    \\& x_{t+k+1|t} = Ax_{t+k|t} + Bu_{t+k|t} + d_{t+k|t}
    \\& u_{t+k|t} = \bar{u}_{t+k|t}
        + {\textstyle\sum_{j=0}^{k-1}}K_{kj|t}d_{t+j|t}
    \\& x_{t+k|t} \in \mathcal{X},\ u_{t+k|t} \in \mathcal{U}
    \\& \forall k \in \{0,1,\dots,N-1\}
    \\& \bar{x}_{t|t} = x(t),\ x_{t|t} = x(t),\ x_{t+N|t} \in \mathcal{O}
    \\& \forall \{d_{t+k|t}\}_{k=0}^{N-1} \subset \mathcal{D}
\end{aligned}~.
\end{equation}
The problem \cref{eq:robust-mpc} optimizes a time-varying feedback policy with a cost on the nominal trajectory $(\bar{x}, \bar{u})$ subject to state and input constraints on the realized trajectory $(x, u)$. We use the subscript $t+k|t$ for quantities at the $k$-th step of the prediction horizon when \cref{eq:robust-mpc} is solved online at time~$t \in \mathbb{N}_{\geq 0}$. If the function $V_N : \mathcal{X} \to \R$, the terminal set~$\mathcal{O}$, and the disturbance set~$\mathcal{D}$ are convex, then~\cref{eq:robust-mpc} is a convex problem; we refer readers to~\cite{goulart_optimization_2006} for implementation details. One might also consider a formulation of \cref{eq:robust-mpc} where the feedback gains are fixed, yielding a more basic tube MPC problem \cite{mayne_robust_2005}.

Solving~\cref{eq:robust-mpc} in a receding horizon fashion encodes the closed-loop feedback policy $u^\star : \mathcal{X} \times \mathbb{N}_{\geq 0} \to \mathcal{U}$, where $u^\star(x(t),t) = u_{t|t}^\star$ is the first element of an optimal control input sequence from solving \cref{eq:robust-mpc} with the initial condition $(x(t),t)$. This feedback policy is \emph{robust} since the constraints in \cref{eq:robust-mpc} are enforced for every possible $N$-step sequence of disturbances.

\subsection{Invariant Sets}
The choice of terminal ingredients $V_N$ and $\mathcal{O}$ in \cref{eq:robust-mpc} is pivotal to guarantee that the closed-loop system formed by the dynamics \cref{eq:problem-dynamics} and the MPC policy satisfies $x(t)\in\calX$, $u(t) \in \calU$ for all $t\geq0$, and is stable. In particular, to establish these recursive feasibility and stability guarantees, the terminal set $\mathcal{O}$ must be invariant with respect to the underlying dynamics \cref{eq:problem-dynamics} under some policy that satisfies the input constraints on $\mathcal{O}$. Moreover, the terminal cost $V_N$ must be a Lyapunov function associated with the stage cost on $\mathcal{O}$ under the policy associated with~$\mathcal{O}$. We review some invariant set notions below.

\begin{definition}[invariant sets \cite{BorrelliBemporadEtAl2017}]
    Consider the dynamical system $x(t+1) = q(x(t), v(t))$, where $v(t)$ is a disturbance signal that takes values in some set $\mathcal{V}$, 
    subject to the state constraint set $\mathcal{X} \subseteq \R^n$. Then a \emph{robust positive invariant (RPI) set} for the constrained system is any set $\mathcal{O} \subseteq \mathcal{X}$ satisfying
    \begin{equation}
        x(0) \in \mathcal{O} \implies x(t) \in \mathcal{O},\ \forall t \in \mathbb{N}_{\geq 0},
    \end{equation}
    for any possible sequence ${\{v(t) \mid t \in \mathbb{N}_{\geq 0}\} \subset \mathcal{V}}$. The \emph{maximal RPI set} $\mathcal{O}_\infty \subseteq \calX$ is the RPI set satisfying ${\mathcal{O} \subseteq \mathcal{O}_\infty}$ for any other RPI set $\mathcal{O} \subseteq \mathcal{X}$.
\end{definition}

Invariant sets are essential to predictive control design, as planning a trajectory into an RPI set associated with a fixed stabilizing feedback policy guarantees that there exists a robust MPC policy that satisfies the constraints for all time. Moreover, computing the maximal RPI set for a linear time-invariant system is algorithmically straightforward \cite{BorrelliBemporadEtAl2017}.

\subsection{ISS Stability} 
Due to the disturbance term $v(t)$, the system \cref{eq:problem-dynamics} typically cannot be regulated to the origin even asymptotically. Therefore, we briefly review relevant results of input-to-state stability (ISS) theory, which is often used to analyze robust control algorithms \cite{LimonAlamoEtAl2009, goulart_optimization_2006, bujarbaruah_semi-definite_2020}. We propose an adaptive  approach that refines an estimate of the unknown function~$f$ online. As a result, the closed-loop system is time-varying. First, we review standard comparison function notation \cite{li_input--state_2018}. A function $\alpha: \R_{\geq 0} \to \R_{\geq 0}$ is a class-$\mathcal{K}$ function if it is continuous, strictly increasing, and $\alpha(0) = 0$. In addition, $\alpha$ is class-$\mathcal{K}_\infty$ if it is class-$\mathcal{K}$ and $\lim_{x\to\infty}\alpha(x) = \infty$. A function $\beta : \R_{\geq0} \times \mathbb{N}_{\geq0 }\to \R_{\geq0}$ is class-$\mathcal{KL}$ if $\beta(\cdot, t)$ is class-$\mathcal{K}$ for any fixed $t\geq0$,  $\beta(x,\cdot)$ is decreasing for any fixed $x\geq 0$, and $\lim_{t\to\infty}\beta(x,t) = 0$ for any fixed $x\geq0$. We now use these function classes to state the ISS definitions.
\begin{definition}[input-to-state stable (ISS) \cite{li_input--state_2018}]\label{def:iss}
    The system ${x(t+1) = q(t, x(t), v(t))}$ with disturbance~$v(t)$ is globally \emph{input-to-state stable (ISS)} if there exists a class-$\mathcal{KL}$ function ${\beta: \R_{\geq 0} \times \mathbb{N}_{\geq 0} \to \R_{\geq 0}}$ and a class-$\mathcal{K}$ function ${\gamma: \R_{\geq 0} \to \R_{\geq 0}}$ such that
    \begin{equation}
        \|x(t)\| \leq \beta(\|x(0)\|, t) + \gamma(\textstyle{\sup_{k \in \{0, 1, \dots, t\}}} \|v(k)\|),
    \end{equation}
    for all $x(0) \in \R^n$ and $t \in \mathbb{N}_{\geq 0}$.
\end{definition}

In essence, a system is ISS if it is nominally asymptotically stable and the influence of the disturbance is bounded. This makes ISS a convenient framework to analyze the stability of systems subject to random disturbances. Similarly to nonlinear stability analysis for deterministic systems, we can show a system is ISS if there exists an ISS-Lyapunov function. 

\begin{definition}[ISS-Lyapunov function \cite{li_input--state_2018}]
    The function $V: \mathbb{N}_{\geq0} \times\R^n \to \R$ is an \emph{ISS-Lyapunov function} for the system ${x(t+1) = q(t, x(t), v(t))}$ if it is continuous in $x$, continuous at the origin for all $t \in \mathbb{N}_{\geq 0}$, and there exist three class-$\mathcal{K}_{\infty}$ functions $\alpha_1,\alpha_2,\alpha_3$ and a class-$\mathcal{K}$ function $\sigma$ such that
    \begin{equation}\begin{aligned}
        \alpha_1(\|x(t)\|) \leq V(t,x(t)) &\leq \alpha_2(\|x(t)\|)
        \\
        V(t+1, x(t+1)) - V(t, x(t)) &\leq -\alpha_3(\|x(t)\|) + \sigma(\|v(t)\|)
    \end{aligned}~,
    \end{equation}
    for all $x(t) \in \R^n$.
\end{definition}

\begin{theorem}[\!\!\label{thm:iss}\cite{li_input--state_2018}]
    A time-varying system is globally ISS if it admits an ISS-Lyapunov function.
\end{theorem}

The above definitions naturally extend to local ISS stability; for a detailed discussion, we refer readers to \cite{li_input--state_2018,jiang_input--state_2001,LimonAlamoEtAl2009}.

\section{Adaptive Robust MPC}\label{sec:approach}
In this section, we first describe assumptions on and necessary features of the learning procedure in a way that is agnostic to the choice of learning algorithm. We then introduce our adaptive robust MPC approach, and prove stability of the combined learning and control framework.

\subsection{Learning Desiderata}
\rev{
Since the nonlinear dynamics term~$f$ is unknown, our method takes a \emph{certainty equivalent} approach by substituting an estimate~$\hat{f}$ that is refined online as more data becomes available. To guarantee the adaptive robust MPC framework satisfies state and input constraints for all time (i.e., safety), we make several assumptions on~$\hat{f}$. We discuss two commonplace estimators that satisfy these assumptions later in \cref{sec:est}.}

\rev{We maintain the estimate
\begin{equation}
    \hat{f}(x,t) = \what(t)\phi(x)
\end{equation}
of $f(x)$, where $\what(t) \in \R^{n \times d}$ is our estimate of $W$ at time~$t$. To this end, we need bounds on our initial uncertainty, i.e., the difference between $f(x)$ and $\hat{f}(x,0)$ for all $x$. For a general statistical estimator, this entails specifying a risk tolerance ${\delta \in (0,1)}$ and computing confidence intervals on the estimate.}

\begin{assumption}[prior knowledge]\label{as:prior-conf}
    Let $w_i$, $\hat{w}_i(t)$, and $\tilde{w}_i(t)$ be the $i$-th rows of $W$, $\what(t)$, and $\widetilde{W}(t) \coloneqq \what(t) - W$, respectively, for $i \in \{1,2,\dots,n\}$. At $t=0$, we know an initial estimate $\what(0) \in \R^{n \times d}$ and bounded sets $\{\mathcal{W}_i(0)\}_{i=1}^n$ with $\mathcal{W}_i(0) \subset \R^d$, such that $\widetilde{W}(0) \in \mathcal{W}(0)$ with probability at least $1 - \delta$, where we define the sets
    \begin{equation}
        \mathcal{W}(t) \coloneqq \big\{
            \widetilde{W} \in \R^{n \times d} \mid \tilde{w}_i \in \mathcal{W}_i(t),\,\forall i \in \{1,\dots,n\}
        \big\},
    \end{equation}
    for all $t \in \mathbb{N}_{\geq 0}$.
\end{assumption}

\rev{\cref{as:prior-conf} provides only an initial bound on the error of the estimate, which we explicitly label as the estimate at $t=0$. Later, we will define $\mathcal{W}(t)$ for all $t \in \mathbb{N}_{\geq 0}$ when we adaptively update our estimate $\what(t)$ and the bounds $\{\mathcal{W}_i(t)\}_{i=1}^n$ online. Our approach leverages the certainty equivalent ``estimate and cancel'' control laws pioneered in classical unconstrained adaptive control \cite{SlotineLi1991, LavretskyWise2013}. As such, \cref{as:bounded-feat} and \cref{as:prior-conf} are necessary to bound our approximation error. However, since we specify the risk tolerance $\delta$, we cannot guarantee exact constraint satisfaction for all time. Instead, we slightly relax our definition of safety.}

\begin{definition}\label{def:safety}
\rev{Under a risk tolerance of $\delta \in [0,1]$, the system \cref{eq:problem-dynamics} is \emph{safe} when 
\begin{equation}\label{eq:safety}
    \mathrm{Prob}\big(x(t) \in \calX,\, u(t) \in \calU,\, \forall t\geq 0\big) \geq 1 - \delta.
\end{equation}
That is, the probability of a constraint violation should be no more than $\delta$ over the \emph{entire} realized trajectory.}
\end{definition}

\rev{To guarantee closed-loop safety, we assume we have an online adaptation strategy that ensures the quality of the estimate $\what(t)$ cannot get worse over time.}

\begin{assumption}[online learning]\label{as:shrinking}
    We have an \emph{online parameter estimator} that maps an initial estimate $\what(0)$, the associated $1 - \delta$ confidence interval $\calW(0)$, and the trajectory history $\{x(k), u(k)\}_{k=0}^t$ to an \emph{online estimate}~$\what(t)$ and confidence interval $\calW(t)$ at time $t$, such that $\hat{w}_i(t) - w_i \in \calW_i(t)$ for all time $t \in \mathbb{N}_{\geq 0}$ and $i \in \{1,2\dots,n\}$ with probability at least $1-\delta$. We assume the confidence intervals on $\what(t)$ are not growing with time, i.e., that%
    \begin{equation}
        \calW(t+1) \subseteq \calW(t),
    \end{equation}
    for all $t \in \mathbb{N}_{\geq 0}$.
\end{assumption}

\rev{Intuitively, \cref{as:shrinking} states that more data should not decrease the confidence in our estimate of $W$. We discuss two commonplace estimators that satisfy \cref{as:shrinking} in~\cref{sec:est}. Formulating separate confidence intervals for each row of~$\what(t)$ is a natural approach, as fitting~$\what(t)$ to historical data decomposes into $n$ separate least-squares problems (one for each row) if the cross-covariance of~$v(t)$ is zero.}
 
\rev{Crucially, \cref{as:shrinking} allows us to treat the confidence intervals $\calW(t)$ as exact bounds in the control design, since a robust controller that guarantees constraint satisfaction conditioned on the event that $\hat{w}_i(t) - w_i \in \calW_i(t)$ for $i \in \{1, \dots, n\}$ and all time then satisfies \cref{eq:safety}. More formally, let $\mathrm{SAFE}$ denote the event that $x(t) \in \calX$ and $u(t) \in \calU$ for all $t\geq 0$. Then, apply the law of total probability to see that
\begin{equation}
\hspace{-0.4em}\begin{aligned}
    &\mathrm{Prob}\big(\mathrm{SAFE})
    \\
    &\geq \mathrm{Prob}\big(\mathrm{SAFE} \  | \ \what(t) - W \in \calW(t), \, \forall t \geq 0 \big) \\
    &\quad \times \mathrm{Prob}\big( \what(t) - W \in \calW(t), \, \forall t \geq 0 \big)
    \\
    &\geq \underbrace{\mathrm{Prob}\big(\mathrm{SAFE} \  | \ \what(t) - W \in \calW(t), \, \forall t \geq 0 \big)}_{=1~\text{using robust MPC in \cref{sec:unmatched}}} (1-\delta)
\end{aligned}\,,
\end{equation}
since application of the estimators in \cref{sec:est} that satisfy \cref{as:shrinking} ensures $\mathrm{Prob}\big(\what(t) - W \in \calW(t), \, \forall t \geq 0 \big) \geq 1 -\delta$.}

\rev{Therefore, we treat the chance constraint \cref{eq:safety} as a proxy for robust constraint satisfaction and construct our approach for the remainder of \cref{sec:approach} conditioned on the event that $\hat{w}_i(t) - w_i \in \calW_i(t)$ for all time and $i\in \{1,\dots, N\}$. This approach was also taken in \cite{lew_safe_2020, DeanManiaEtAl2018}.  }

\rev{\cref{as:shrinking} does not require the estimated range of the unknown function $f$ to shrink over time. Therefore, our estimation procedure differs from methods such as \cite{bujarbaruah_adaptive_2018, bujarbaruah_semi-definite_2020, soloperto_learning-based_2018} that refine a non-increasing bound exclusively on the range of~$f$ without taking direct advantage of the structure in the nonlinear dynamics.}

\subsection{Certainty Equivalent Cancellation}\label{sec:unmatched}
We propose optimizing over feedback policies that cancel as much of the nonlinear term $f(x)$ as possible.
\begin{definition}
The set of \emph{matching certainty equivalent (CE) policies} is the time-varying function class whose elements $\pi: \calX \times \mathbb{N}_{\geq 0} \to \calU$ are of the form
\begin{equation}\label{eq:unmatched-ce-law}
    \pi(x(t), t) = u^\star(x(t), t) - B^\dagger \hat{f}(x(t), t)
    \smallskip
\end{equation}
\end{definition}
The matching CE policies simply project $\hat{f}(x(t), t)$ onto $\operatorname{Range}(B)$, and cancel out as much of the disturbance as possible in the Euclidean norm sense, since
\begin{equation}
    B^\dagger \hat f(x(t),t) = \argmin_z \|Bz - \hat{f}(x(t), t)\|.
\end{equation}
The matching CE law \cref{eq:unmatched-ce-law} results in the closed-loop dynamics%
\begin{equation}\label{eq:unmatched-error}
\begin{aligned}
    x(t+1) 
    &= Ax(t) + B\pi(x(t),t) + f(x(t)) + v(t) \\
    &= Ax(t) + Bu^\star(x(t),t) + d(t)
\end{aligned}~,
\end{equation}
where we define the compound disturbance term $d(t)$ as
\begin{equation}\begin{aligned}
    d(t) 
    &\coloneqq v(t) + f(x(t)) - BB^\dagger\hat{f}(x(t),t) \\
    &= v(t) \begin{aligned}[t]
        &+ BB^\dagger(f(x(t)) - \hat{f}(x(t),t)) \\
        &+ (I - BB^\dagger)f(x(t))
    \end{aligned}
\end{aligned}~.
\end{equation}
We have written $d(t)$ above with three terms to highlight that it is driven by the process disturbance~$v(t)$, the estimation error~$f(x) - \hat{f}(x)$, and the imperfect matching using~$B^\dagger$.

\begin{remark}
If we know that $f$ is a \emph{matched} uncertainty (i.e., that $f(x(t)) = Bg(x(t))$ for some function $g : \R^n \to \R^m$), we can reduce the compound disturbance to
\begin{equation}
\begin{aligned} \label{eq:matched-dist-dyn}
    d(t) &\coloneqq v(t) + B(g(x(t)) - \hat{g}(x(t),t)),
\end{aligned}
\end{equation}
since $B^\dagger B = I$. Therefore, the matching certainty equivalent controller \cref{eq:unmatched-ce-law} generalizes approaches for systems with matched uncertainty to those with unmatched uncertainty.
\end{remark}

\rev{As mentioned in our learning desiderata, to guarantee safety as defined in \cref{eq:safety}, we construct a controller that guarantees constraint satisfaction when $\what(t) - W \in \calW(t)$ for all $t\geq 0$, which we will assume is the case for the remainder of this section. To design a robust matching CE policy, we need to: 
\begin{itemize}
    \item Guarantee that the closed-loop dynamics \cref{eq:unmatched-error} do not violate state constraints.
    \item Tighten input constraints to account for the certainty equivalent cancellation in \cref{eq:unmatched-ce-law}. 
\end{itemize}
Therefore, we introduce two simple polytopic approximations that bound the support of the cancellation term in \cref{eq:unmatched-ce-law} and the terms that make up the compound disturbance $d(t)$.}


\begin{lemma}\label{lem:unmatched-support}
Consider online approximation of $f(x)$ with features satisfying \cref{as:bounded-feat}, an estimator satisfying \cref{as:shrinking}, and define the estimated support set as
\begin{equation}
    \calF(t) \coloneqq \{z \in \R^n : |z_i| \leq \|\hat{w}_i(t)\| + 2 \max_{\tilde{w}_i \in \calW_i(t)}\|\tilde{w}_i\|\}.
\end{equation}
Then, for all $x \in \mathcal{X}$ and $t,k \in \mathbb{N}_{\geq 0}$ it holds that
\begin{equation}
    \hat{f}(x, t + k), f(x) \in \calF(t).
\end{equation}
\end{lemma}
\begin{proof}
We show $\calF(t)$ over-approximates the range of values $z =\hat{f}(x, t+k)$ can take for any $x \in \calX$ and $k \geq 0$. Let $z = \hat{f}(x, t+k) = \widehat{W}(t+k)\phi(x)$. Then,
\begin{equation}\begin{aligned}
    |z_i|
    &= |\hat{w}_i(t+k)^\top\phi(x)| \\
    &\leq \|\hat{w}_i(t+k)\| \\
    &\leq \|\hat{w}_i(t)\| + \|\hat{w}_i(t+k) - \hat{w}_i(t)\| \\
    &\leq \|\hat{w}_i(t)\| + \|\hat{w}_i(t+k) - w_i\| + \|\hat{w}_i(t) - w_i\|.
\end{aligned}
\end{equation}
The non-increasing confidence interval property from \cref{as:shrinking} gives $ \calW_i(t+k)\subseteq \calW_i(t)$ for $k\geq 0$, so
\begin{equation}
    |z_i| \leq \|\hat{w}_i(t)\| + 2\max_{\tilde{w}_i \in \calW_i(t)}\|\tilde{w}_i\|,
\end{equation}
which proves that~$\hat{f}(x, t+k) \in \calF(t)$ for all~$k\geq 0$. In addition, let~$y = f(x) = W\phi(x)$ for some $x\in\calX$. Then, 
\begin{equation}\begin{aligned}
    |y_i| 
    &= |w_i^\top \phi(x)| \\
    &\leq \|w_i\| \\
    &\leq \|\hat w_i(t)\| + \|\hat w_i(t) - w_i\| \\
    &\leq \|\hat w_i(t)\| + \max_{\tilde{w}_i \in \calW_i(t)}\|\tilde{w}_i\|.
\end{aligned}
\end{equation}
Hence,~$f(x) \in \calF(t)$. 
\end{proof}

The set $\mathcal{F}(t)$ in \cref{lem:unmatched-support} contains all possible values that our online estimate can take for all future times. It is not straightforward to create a tighter approximation (i.e., eliminate the factor of 2) without additional assumptions. To see this, consider a constant unit norm ball confidence interval. In the worst case, the true parameter lies on the boundary of the ball around the current estimate. This means all future estimates may lie a Euclidean distance of 2 units away from the current estimate, yielding the bound in \cref{lem:unmatched-support}.

\begin{remark}
Representing the set $\calF(t)$ defined in \cref{lem:unmatched-support} requires finding the max-norm element of a convex set. For many convex confidence intervals this is straightforward to compute, although it is not always efficient. For ellipsoids, this amounts to an eigenvalue computation. For polytopes, this requires vertex enumeration, incurring exponential complexity in the dimension of the state space $n$.
\end{remark}

Moreover, \cref{lem:unmatched-support} does not require that $\calF(t+1) \subseteq \calF(t)$, so we provide the following corollary to help us create an approximation that is non-increasing in size.
\begin{corollary}\label{cor:unmatched-support}
    At time $t$, the sets $\{\calF(i)\}_{i=0}^t$ are known, so
    \begin{equation}
        \hat{f}(x, t + k), f(x) \in \bigcap_{i=0}^t \calF(i) \eqqcolon \calFhat(t)
    \end{equation}
    for all $x \in \calX$ and $k \in \mathbb{N}_{\geq 0}$, where we define $\calFhat(t)$ as the set
    \begin{equation}\label{eq:calfhat-def}
        \{z : |z_i| \leq \min_{j \in \{0,\dots,t\}}[\|\hat{w}_i(j)\| 
            + 2\!\max_{\tilde{w}_i \in \calW_i(j)}\|\tilde{w}_i\|]\}.
    \end{equation}
    Note $\calFhat(t)$  can be computed recursively in time.
\end{corollary}
To construct a robust MPC problem to optimize the CE policy~\cref{eq:unmatched-ce-law}, we need to account for the compound disturbance~$d(t)$. We do this with the following lemma.
\begin{lemma}\label{lem:unmatched-disturbance}
Assume the online parameter estimator satisfies \cref{as:shrinking} with features that satisfy \cref{as:bounded-feat} and define the approximation error support as the set
\begin{equation}
    \mathcal{D}(t) \coloneqq \{z \in \R^n \mid |z_i| \leq \max_{\tilde{w}_i \in \calW_i(t)} \|\tilde{w}_i\|,  \forall i \in \{1, \dots,n\}\}.
\end{equation}
If we control the system \cref{eq:problem-dynamics} using the certainty equivalent control law \cref{eq:unmatched-ce-law}, then at time $t$ for all $k\in \mathbb{N}_{\geq 0}$, the compound disturbance $d(t+k)$ in the dynamics~\cref{eq:unmatched-error} is contained in the set $\calDhat(t) \subseteq \calDhat(t-1)$, defined as
\begin{equation}\label{eq:unmatched-dist}
    \calDhat(t) \coloneqq (I - BB^\dagger)\calFhat(t) \oplus BB^\dagger \calD(t) \oplus \calV.
\end{equation}
Here $\oplus$ indicates the Minkowski sum and a matrix-set multiplication indicates a linear transformation of the set's elements.
\end{lemma}
\begin{proof}
At any state~$x \in \calX$ and time $t + k$, let $z = \hat{f}(x, t+k) - f(x) = (\what(t+k) - W)\phi(x) = \wtilde \phi(x)$ for some $\wtilde \in \calW(t+k)\subseteq \calW(t)$. Then
\begin{align}
    |z_i| &= |\tilde{w}_i^\top \phi(x_t)| \leq \|\tilde{w}_i\| \leq \max_{\tilde{w}_i \in \calW_i(t)} \|\tilde{w_i}\|. \nonumber
\end{align}
So $\hat{f}(x, t+k) - f(x) \in \mathcal{D}(t)$ for all $k \in \mathbb{N}_{\geq 0}$. Since $\calW_i(t) \subseteq \calW_i(t-1)$ by \cref{as:shrinking}, this implies $\mathcal{D}(t) \subseteq \mathcal{D}(t-1)$. Then, by \cref{cor:unmatched-support}, $f(x) \in \calFhat(t)$. Therefore,~$d(t)$ in the closed-loop dynamics \cref{eq:unmatched-error} is contained in $\calDhat(t)$ since $\mathcal{D}(t)$ is symmetric. In addition, since both $\calFhat(t) \subseteq \calFhat(t-1)$ and~$\calD(t) \subseteq \calD(t-1)$, the support of $d(t)$ is nested over time, i.e., $\calDhat(t) \subseteq \calDhat(t-1)$.
\end{proof}

\begin{remark}
If the estimation error is small, \cref{eq:unmatched-dist} shows that the compound disturbances $d(t)$ will only depend on the process noise and the component of the unknown function $f$ that cannot be cancelled using the CE control law \cref{eq:unmatched-ce-law}. \rev{Therefore, if the estimation error is small and a significant component of the uncertainty is matched, we generally expect our method to tolerate nonlinear uncertainty with larger magnitude than a conventional robust MPC scheme.}
\end{remark}

\begin{remark}\label{rem:proj}
We could consider multiple variations on the bounds in \cref{lem:unmatched-support} and \cref{lem:unmatched-disturbance} that would yield equivalent properties of the closed-loop system. 
For example, if a bound on the true range of $f$ is known a priori, we may project $\hat{f}$ into a known box enclosing the support of $f$. 
\end{remark}

\begin{remark}\label{rem:axis-aligned}
\cref{lem:unmatched-support} and \cref{lem:unmatched-disturbance} result in box constraints on the disturbances. Approximations of this form are often more favorable from a practical perspective compared to (operator) norm type bounds, since these box constraints are better able to preserve the relative magnitudes of state variables that represent physical quantities.
\end{remark}

We now use the set $\mathcal{F}(t)$ from \cref{lem:unmatched-support} and the set $\calDhat(t)$ from \cref{lem:unmatched-disturbance} to modify the robust MPC problem \cref{eq:robust-mpc} into
\begin{equation}\label{eq:unmatched-mpc}
\hspace{-1.5em}\begin{aligned}
    \minimize_{\substack{
        \{K_{kj|t}\}_{k=0,j=0}^{N-1,k-1},\\
        \{\bar{u}_{t+k|t}\}_{k=0}^{N-1}
    }}\enspace
    &V_N(\bar{x}_{t+N|t}) + \sum_{k=0}^{N-1} h(\bar{x}_{t+k|t}, \bar{u}_{t+k|t})
    \\
    \subjectto\enspace
    & \bar{x}_{t+k+1|t} = A\bar{x}_{t+k|t} + B\bar{u}_{t+k|t}
    \\& x_{t+k+1|t} = Ax_{t+k|t} + Bu_{t+k|t} + d_{t+k|t}
    \\& u_{t+k|t} = \bar{u}_{t+k|t}
        + {\textstyle\sum_{j=0}^{k-1}}K_{k,j|t}d_{t+j|t}
    \\& x_{k|t} \in \mathcal{X},\ u_{t+k|t} \in \mathcal{U} \ominus B^\dagger\calFhat(t)
    \\& \forall k \in \{0,1,\dots,N-1\}
    \\& \bar{x}_{t|t} = x(t),\ x_{t|t} = x(t),\ x_{t+N|t} \in \mathcal{O}(t)
    \\& \forall \{d_{t+k|t}\}_{k=0}^{N-1} \subset \widehat{\mathcal{D}}(t)
\end{aligned}~.
\end{equation}
Compared to \cref{eq:robust-mpc}, in \cref{eq:unmatched-mpc} we have tightened the input constraints to account for the matching term in the certainty equivalent policy \cref{eq:unmatched-ce-law}. As in standard robust MPC, we assume we can compute a robust control invariant set $\mathcal{O}(t)$ and that we have access to a convex terminal cost function $V_N$.
 
Let $\{u^\star_{t|t}, \dots, u^\star_{t+N-1|t}\}$ be an optimal policy sequence for~\cref{eq:unmatched-mpc}. Online, we solve \cref{eq:unmatched-mpc} at each time step $t \in \mathbb{N}_{\geq 0}$ and choose the robust control term of the \emph{certainty equivalent} control policy \cref{eq:unmatched-ce-law} as the receding horizon feedback law ${u^\star : \mathcal{X} \times \mathbb{N}_{\geq 0} \to \mathcal{U}}$ such that 
\begin{equation}\label{eq:unmatched-robust-pol}
    u^\star(x(t),t) = u_{t|t}^\star.
\end{equation}

\begin{assumption}\label{as:unmatched-termcost}
    The terminal cost $V_N: \R^n \to \R_+$ is a continuous convex Lyapunov function for the nominal dynamics under a policy $u_N(x) = - Kx$. That is, there exists a class-$\mathcal{K}_\infty$ function $\alpha_N(\|x\|) \geq h(x, u_N(x))$ for which
    \begin{equation}
        V_N((A- BK)x) - V_N(x) \leq -\alpha_N(\|x\|),
    \end{equation}
    for all $x \in \R^n$.
\end{assumption}

\begin{assumption}\label{as:unmatched-termconst}
For the policy $u_N(x) = - Kx$ in \cref{as:unmatched-termcost}, the terminal set $\mathcal{O}(t) \subseteq \mathcal{X}$ is a maximal robust positive invariant set for the closed-loop system 
$x(k+1) = (A - BK)x(k) + d(k)$ for $d(k) \in \calDhat(t)$ subject to $x(k) \in \calX$ and $u_N(x(k)) \in \calU \ominus B^\dagger \calFhat(t)$ for all $t\geq 0$. 
\end{assumption}

\cref{as:unmatched-termcost} and \cref{as:unmatched-termconst} are standard and easily satisfied by taking $u_N(x) = -Kx$ and $V_N(x) = x^\top P x$ as the solution to an LQR problem with $h$ as the stage cost. Then, $\mathcal{O}(t)$ can be computed efficiently using the standard algorithms in \cite{BorrelliBemporadEtAl2017}. 

\Cref{lem:unmatched-disturbance,lem:unmatched-support} imply that $\mathcal{O}(t-1) \subseteq \mathcal{O}(t)$ since $\calDhat(t) \subseteq \calDhat(t-1)$ and $\calFhat(t) \subseteq \calFhat(t -1)$. Therefore, the terminal constraint becomes less conservative over time. 

\subsection{Stability}\label{sec:properties}
We prove the stability of our algorithm through a recursive feasibility and input-to-state stability argument.
\begin{theorem}\label{thm:unmatched-recfeas}
Consider the system \cref{eq:problem-dynamics}, a parameter estimator that satisfies \cref{as:shrinking} with features that satisfy \cref{as:bounded-feat} in closed-loop feedback with the \emph{matching certainty equivalent} control law \cref{eq:unmatched-ce-law},\cref{eq:unmatched-robust-pol}. If the tube MPC problem \cref{eq:unmatched-mpc} is feasible at~$t=0$, then for all~$t\geq 0$ we have that \cref{eq:unmatched-mpc} is feasible and the closed-loop system \cref{eq:problem-dynamics},\cref{eq:unmatched-ce-law},\cref{eq:unmatched-robust-pol} satisfies~$x(t) \in \calX$, and $\pi(x(t),t) \in \calU$.
\end{theorem}
\begin{proof}
Suppose the optimal control problem \cref{eq:unmatched-mpc} is feasible at time $t$, with solution~$\{u^\star_{t|t}(\cdot), \dots, u^\star_{t+N-1|t}(\cdot)\}$. By \cref{lem:unmatched-support} and \cref{cor:unmatched-support} we have that $\hat{f}(x, t +k) \in \calFhat(t)$ for all $k \in \mathbb{N}_{\geq 0}$. Therefore, the time-varying CE control law
\begin{align}\label{eq:proof-pol}
    \pi_{t+k|t}(\cdot) = u^\star_{t+k|t}(\cdot) - B^\dagger \hat{f}(x(t), t + k)
\end{align}
satisfies the input constraints for $t \in [t, t+N-1]$, since \cref{eq:unmatched-mpc} then implies $u^\star_{t+k|t}(\cdot) \in \calU \ominus B^\dagger \calFhat(t)$. Moreover, by \cref{lem:unmatched-disturbance} the disturbance support does not grow in time, i.e., $\calDhat(t+1) \subseteq \calDhat(t) $. Therefore, we have that under policy \cref{eq:proof-pol} the closed-loop trajectory formed by \cref{eq:problem-dynamics},\cref{eq:proof-pol} satisfies $x(t+k) \in \calX$ for all $k \in [0,N]$ and that $x(t+N) \in \mathcal{O}(t)$. Hence, if we apply the CE policy \cref{eq:unmatched-ce-law},\cref{eq:unmatched-robust-pol} at time $t$, then $x(t+1) \in \calX$ and $\pi(x(t), t) \in \calU$.

By \cref{as:unmatched-termconst}, for any $x \in \mathcal{O}(t) \subseteq \mathcal{O}(t+1)$, applying the policy $u_N(x) \in \calU \ominus B^\dagger \calFhat(t)$ implies that $Ax + Bu_N(x) + d \in \mathcal{O}(t+1)$ for any $d \in \calDhat(t)$. Therefore, \cref{cor:unmatched-support} and \cref{lem:unmatched-disturbance} imply the policy sequence $\{u^\star_{t+1|t}(\cdot), \dots u^\star_{t+N-1|t}(\cdot), u_N(\cdot)\}$ is feasible for the tube MPC problem \cref{eq:unmatched-mpc} at time $t+1$. Therefore, if the MPC program \cref{eq:unmatched-mpc} is feasible at time $t=0$, it is also feasible for all $t\geq0$ and the closed-loop system formed by the matching CE law \cref{eq:problem-dynamics}, \cref{eq:unmatched-ce-law}, \cref{eq:unmatched-robust-pol} must robustly satisfy state and input constraints by induction.
\end{proof}

\begin{theorem}\label{thm:unmatched-iss}
Consider a system of the form in \cref{eq:problem-dynamics}, 
a parameter estimator that satisfies \cref{as:shrinking} with features that satisfy \cref{as:bounded-feat} in closed-loop feedback with the \emph{certainty equivalent} control law \cref{eq:unmatched-ce-law},\cref{eq:unmatched-robust-pol}. Let $\calX_N \subseteq \calX$ denote the set of states for which the tube MPC problem \cref{eq:unmatched-mpc} is feasible. Then the closed-loop system is locally input-to-state stable with region of attraction $\calX_N$. 
\end{theorem}
\begin{proof}
Our proof closely resembles \cite[Thm.~2]{BujarbaruahZhangEtAl2020}. 
We argue that the nominal system is stable by a standard MPC argument, and that the closed-loop system is ISS since the disturbances are bounded. Let $J^\star_N(t, x(t))$ be the optimal value of~\cref{eq:unmatched-mpc} associated with the nominal prediction $\{\bar{x}_{t|t}^\star, \dots, \bar{x}_{t+N|t}^\star\}$ and feedback policies $\{u^\star_{t|t}(\cdot), \dots, u^\star_{t+N-1|t}(\cdot)\}$. Since we assume the stage cost is quadratic, there exist two class-$\mathcal{K}_\infty$ functions $\alpha_1, \alpha_2$ such that for all $t\geq0$, $\alpha_1(\|x\|) \leq J^\star_N(t,x) \leq  \alpha_2(\|x\|)$, a class-$\mathcal{K}_\infty$ function $\alpha_3$ such that $h(x,u) \geq \alpha_3(\|x\|)$, and  $J_N^\star(t,0)=0$ (see \cite[Prop.~1]{LimonAlamoEtAl2009}, \cite[Thm.~2]{BujarbaruahZhangEtAl2020}). As in the proof of \cref{thm:unmatched-recfeas}, we have that if we apply the CE control law \cref{eq:unmatched-ce-law},\cref{eq:unmatched-robust-pol} at time $t$, then the policies $\{u^\star_{t+1|t}, \dots, u^\star_{t+N-1|t}, u_N\}$ are a feasible solution for \cref{eq:unmatched-mpc} at time $t+1$. Let $\bar{J}(t, x)$ be the cost associated with forward simulating the nominal system using the policies $\{u^\star_{t+1|t}, \dots, u^\star_{t+N-1|t}, u_N\}$ with $x$ as initial condition.  i.e., set $u^\star_{t+N|t} = u_N(\cdot)$ and let $\bar{x}_{t+1|t+1} = x$, $\bar{x}_{k+1|t+1} = A\bar{x}_{k|t+1} + Bu^\star_{k|t}(\bar{x}_{k|t+1})$ for $k \in \{t+1, \dots, t+N\}$ so that
\begin{equation}
    \bar{J}(t,x) = \sum_{k=t+1}^{t + N}h(\bar{x}_{k|t+1}, u^\star_{k|t}(\bar x_{k|t+1})) + V_N(\bar{x}_{t+N+1|t+1}). \nonumber
\end{equation}
This gives that $J^\star_N(t+1, x(t+1)) \leq \bar{J}(t, x(t+1))$. Moreover, since the stage cost is quadratic and by \cref{as:unmatched-termcost}, $\bar{J}(t,x)$ is uniformly continuous in $x$ for all $t\geq 0$ on the state space since the inputs are constrained in a compact set. It follows that for $x_1, x_2 \in \calX$, there exists a $\mathcal{K}_\infty$ function $\alpha_J$ such that for all $t\geq0$, $|\bar{J}(t,x_1) - \bar{J}(t,x_2)| \leq \alpha_J(\|x_1 - x_2\|)$ (see \cite[Lem.~1]{LimonAlamoEtAl2009}). Therefore, 
\begin{equation}\begin{aligned}
        &J^\star_N(t + 1, x(t+1)) - J^\star_N(t, x(t)) 
        \\
        &\leq \bar{J}(t, x(t+1)) - J^\star_N(t, x(t)) 
        \\
        &= \bar{J}(t, x(t+1)) - \bar{J}(t, \bar{x}^\star_{t+1|t}) + \bar{J}(t, \bar{x}^\star_{t+1|t}) - J^\star_N(t, x(t)) 
        \\
        &\leq |\bar{J}(t, x(t+1)) - \bar{J}(t, \bar{x}^\star_{t+1|t})|
            - h(x_t, \bar{u}^\star_{t|t}(x(t))) 
        \\
        &\leq \alpha_J(\|d(t)\|) - \alpha_3(\|x(t)\|).
\end{aligned} \nonumber
\end{equation}
So the system is ISS by \cref{thm:iss}.
\end{proof}
\begin{remark}
The ISS result in \cref{thm:unmatched-iss} does not explicitly show that improvements in the confidence of the model lead to better performance of the controller, since we only assume the model confidence is non-decreasing in \cref{as:shrinking}. 
\end{remark}
In adaptive control, stronger guarantees of performance improvement are typically made under \emph{persistence of excitation} assumptions \cite{SlotineLi1991}.

\begin{remark}\label{rem:exogenous}
\rev{The only property of the feature map $\phi$ that we relied on to prove Lemmas \ref{lem:unmatched-support}, \ref{lem:unmatched-disturbance} was the boundedness assumption $\|\phi(x) \| \leq 1$. Therefore, Lemmas \ref{lem:unmatched-support}, \ref{lem:unmatched-disturbance} and \cref{cor:unmatched-support} also hold when the feature map depends on an additional \emph{exogenous} signal $z(t) \in \mathcal{Z} \subseteq \R^p$ beside the system state $x(t)$, as long as $\|\phi(x, z)\| \leq 1$ for all $x\in\calX$ and $z \in  \mathcal{Z}$. This implies that Theorems \ref{thm:unmatched-recfeas} and \ref{thm:unmatched-iss} also hold when the bounded feature map $\phi$ depends on an exogenous signal $z(t)$. Therefore, we can easily incorporate observable time-varying environmental effects that we know influence the dynamics, even when it is unclear how these exogenous signals evolve over time, reminiscent of approaches in parameter-varying control \cite{balas_lpv}. We take this approach on the example in \cref{sec:exp-cruise}.}
\end{remark}
\subsection{Iterative/Episodic Learning}\label{sec:episodic}

\rev{Oftentimes, a controller is used to perform the same task repeatedly over many iterations (or episodes). We can apply the algorithm we developed in \cref{sec:unmatched} to such an iterative learning (IL), or episodic, setting as long as the tube MPC problem \cref{eq:unmatched-mpc} is feasible at the start of each iteration. In this subsection, we take advantage of the properties of the approach presented in \cref{sec:unmatched} to reduce the amount of online computation required and update some of the MPC problem parameters only offline between episodes, generalizing our method to a \emph{family} of algorithms suited to episodic operation. }

\begin{definition}
\rev{In an \emph{iterative} (or \emph{episodic}) setting, the system starts from the fixed initial condition $x^{(j)}(0) = x(0) \in\calX$ at $t=0$ and evolves for a finite duration $T^{(j)} \in \mathbb{N}_{> 0}$ for each iteration $j\in\mathbb{N}_{\geq 0}$. We use the notation $x^{(j)}(t)$ to denote the state at time $t$ of iteration $j$, and apply this notation to time-varying quantities in general, e.g., we use the notation $u^{(j)}(t)$ for the input into the system.}
\end{definition}

\rev{Updating the sets $\calFhat$,~$\calDhat$, and in particular  $\mathcal{O}$ online can require significant computation. Therefore, it is attractive to only update these sets offline between iterations in an episodic setting, depending on available computational resources. We summarize the complete online method we developed in \cref{sec:unmatched} as Algorithm \ref{alg:ce-mpc}.A and introduce two episodic variants as Algorithms \ref{alg:ce-mpc}.B-C that reduce online computation by keeping $\mathcal{O}$ fixed throughout an iteration. In particular,
\begin{itemize}
    \item \textbf{Alg.~\ref{alg:ce-mpc}.B} updates $\calFhat^{(j)}(t)$,~$\calDhat^{(j)}(t)$ at each timestep $t$ online and keeps $\mathcal{O}^{(j)}(t)$ fixed as $\mathcal{O}^{(j)}(0)$ during the $j$-th episode.
    \item\textbf{Alg.~\ref{alg:ce-mpc}.C} keeps $\calFhat^{(j)}(t)$,~$\calDhat^{(j)}(t)$, and $\mathcal{O}^{(j)}(t)$ fixed as $\calFhat^{(j)}(0)$,~$\calDhat^{(j)}(0),$ and $\mathcal{O}^{(j)}(0)$ during the $j$-th episode. 
\end{itemize}
It is straightforward to show that we retain the recursive feasibility and stability guarantees in \cref{thm:unmatched-recfeas} and \cref{thm:unmatched-iss} when we apply the episodic variants of Algorithm \ref{alg:ce-mpc}.A. }

\begin{algorithm}[t]\label{alg:ce-mpc}
\begin{algorithmic}[1]
\REQUIRE \rev{initial estimate $\what(0)$}
\REQUIRE \rev{initial $1-\delta$ confidence interval $\calW(0)$}
\REQUIRE \rev{terminal cost $V_N$ and feedback gain $K$ as in \cref{as:unmatched-termcost}}
\STATE \rev{Initialize $\calFhat$, $\calDhat$, $\mathcal{O}$ using \cref{eq:calfhat-def}, \cref{eq:unmatched-dist}, and e.g., \cite[Alg. 10.4]{BorrelliBemporadEtAl2017}.}
\FOR{\rev{$j = 0, 1, \dots, \mathrm{episodes}$}}{
    \STATE \rev{Set initial condition $x^{(j)}(0) = x(0)$.}
    \FOR{\rev{$t = 0, \dots, T^{(j)}$}}{
        \STATE \rev{Compute $u^\star$ as \cref{eq:unmatched-robust-pol} using \cref{eq:unmatched-mpc}.}
        \STATE \rev{Apply the matching CE policy $\pi$ \cref{eq:unmatched-ce-law}.}
        \STATE \rev{Observe the next state $x^{(j)}(t+1)$.}
        \STATE \rev{Update estimate $\hat{f}$ and confidence interval $\calW$.}
        \STATE \rev{ 
            \textbf{Alg. \ref{alg:ce-mpc}.A:} Update $\{\calFhat$, $\calDhat$, $\mathcal{O}\}$.\\
            \textbf{Alg. \ref{alg:ce-mpc}.B:} Update $\{\calFhat, \calDhat\}$.\\ 
            \textbf{Alg. \ref{alg:ce-mpc}.C:} Continue.}
    }\ENDFOR
    \STATE \rev{
        \textbf{Alg. \ref{alg:ce-mpc}.A:} Continue.\\
        \textbf{Alg. \ref{alg:ce-mpc}.B:} Update $\mathcal{O}$.\\
        \textbf{Alg. \ref{alg:ce-mpc}.C:} Update $\{\calFhat$, $\calDhat$, $\mathcal{O}\}$ using $\{\hat{f}^{(j)}(\cdot, t), \calW^{(j)}(t)\}_{t=0}^{T^{(j)}}$.}
}\ENDFOR
\end{algorithmic}
\caption{Adaptive CE MPC with online or episodic updates}
\end{algorithm}

\begin{lemma}\label{lem:episodic}
\rev{Consider episodic control of the system \cref{eq:problem-dynamics} in closed-loop with Algorithm \ref{alg:ce-mpc}.B or Algorithm \ref{alg:ce-mpc}.C from the initial condition $x(0)$. If the tube MPC problem \cref{eq:unmatched-mpc} is feasible at $t = 0$ at iteration $j=0$, then for all $t\geq 0$ at iterations $j \geq 0$ it holds that \cref{eq:unmatched-mpc} is feasible, the closed-loop system satisfies $x^{(j)}(t) \in \calX$ and $\pi^{(j)}(x^{(j)}(t),t) \in \calU$, and the optimal cost functions $J^{\star,B,(j)}_N(t, x(t))$ and $J^{\star, C,(j)}_N(t, x(t))$ for \cref{eq:unmatched-mpc} are ISS-Lyapunov functions for the closed-loop systems formed by Algorithm \ref{alg:ce-mpc}.B  and Algorithm \ref{alg:ce-mpc}.C  respectively.}
\end{lemma}
\begin{proof}
\rev{First, we note that the proofs of \cref{thm:unmatched-recfeas} and \cref{thm:unmatched-iss} only relied on the facts that $\calDhat^{(j)}(t), \calFhat^{(j)}(t)$ do not grow over time during an iteration $j$, i.e., that $\calFhat^{(j)}(t+1) \subseteq \calFhat^{(j)}(t)$ and $\calDhat^{(j)}(t+1) \subseteq \calDhat^{(j)}(t)$, and that $\mathcal{O}^{(j)}(t)$ does not shrink, i.e., that $\mathcal{O}^{(j)}(t) \subseteq \mathcal{O}^{(j)}(t+1)$. }

\rev{Therefore, if we keep $\calFhat^{(j)}(t) := \calFhat^{(j)}(0)$ and $\calDhat^{(j)}(t) := \calDhat^{(j)}(0)$ fixed over iteration $j$, consequently fixing $\mathcal{O}^{(j)}(t) := \mathcal{O}^{(j)}(0)$, \cref{thm:unmatched-recfeas} and \cref{thm:unmatched-iss} still hold, since \cref{cor:unmatched-support} implies that $f(x^{(j)}(t)), \hat{f}^{(j)}(x^{(j)}(t),t) \in \calFhat^{(j)}(0)$ and  \cref{lem:unmatched-disturbance} implies that $d^{(j)}(t) \in \calDhat^{(j)}(0)$ for all $t\geq 0$. Thus, Alg. \ref{alg:ce-mpc}.C is recursively feasible and ISS over an iteration $j$. }

\rev{Moreover, note that $\mathcal{O}^{(j)}(0)$ is an RPI set for the closed-loop system in \cref{as:unmatched-termconst} associated with disturbances in $\calDhat^{(j)}(0)$, state constraints $\calX$, and input constraints $\calU \ominus B^\dagger \calFhat^{(j)}(0)$. Since $\calFhat$ and $\calDhat$ do not grow over time, it follows that $\mathcal{O}^{(j)}(0)$ is also RPI for the closed-loop system in \cref{as:unmatched-termconst} with disturbances in $\calDhat^{(j)}(t)$ and input constraints $\calU \ominus B^\dagger \calFhat^{(j)}(t)$ for all $t\geq 0$. Hence, the proofs of \cref{thm:unmatched-recfeas} and \cref{thm:unmatched-iss} also apply to Algorithm \ref{alg:ce-mpc}.B.}

\rev{Finally, note that updates to $\calFhat$, $\calDhat$, and $\mathcal{O}$ can only increase the feasible domain of \cref{eq:unmatched-mpc}. Therefore, if \cref{eq:unmatched-mpc} is feasible for Algorithms \ref{alg:ce-mpc}.B-C at $t = 0$ and $j=0$, it will be also be feasible at $t=0$ for all $j\geq 0$ since $x^{(j)}(0) = x(0)$ is constant. This proves the lemma.}
\end{proof}

\rev{\cref{lem:episodic} shows that we can safely reduce the amount of online computation required to apply our adaptive MPC in an episodic setting, even though we still adapt the model online in Algorithms \ref{alg:ce-mpc}.A-C.  We can typically update the model estimate efficiently using recursive filters such as those we discuss in \cref{sec:est}, though one could trivially keep the model fixed over an iteration as well. }

\begin{remark}
\rev{In general, we should expect the computational benefits of applying Algorithms \ref{alg:ce-mpc}.B-C to come at the expense of conservatism, since updates to $\calFhat$, $\calDhat$, and $\mathcal{O}$ can increase the size of the feasible set of problem \cref{eq:unmatched-mpc}. However, it is not straightforward to mathematically relate the realized closed-loop costs when we apply Algorithms \ref{alg:ce-mpc}.A-C to each other, since the model estimates depend on the trajectory histories induced by the applied controller.}
\end{remark}

\rev{We emphasize that under our learning desiderata, episodic application of our robust adaptive MPC in combination with an estimator that satisfies \cref{as:shrinking} ensures that we satisfy the safety guarantee \cref{eq:safety} with probability $1 - \delta$ jointly for all episodes. For some estimators, such as the estimator we discuss in \cref{sec:blr}, one could consider re-initializing the initial estimate at each episode using all previously collected data to yield tighter confidence intervals, thereby reducing conservatism. However, by doing so, one can only guarantee that state and inputs are satisfied with probability at least $1-\delta$ for each episode, as opposed to across all episodes.}

  \section{Adaptation Laws \& Learning Algorithms}\label{sec:est}
In this section we describe two common online function approximation schemes---one statistical and one not---that satisfy the decaying confidence interval of \cref{as:shrinking} that we used to construct our ARMPC algorithm.

\subsection{Set Membership Estimation} \label{sec:setmembership}
A common approach in the adaptive MPC literature is to estimate constant, or slowly changing, disturbances through \emph{set membership} estimation \cite{milanese1991optimal, milanese2013bounding}. These estimators maintain a feasible parameter set that is refined as more data becomes available. The feasible parameter set contains all credible model parameters that explain previous observations, which means that the feasible parameter sets are nested over time. We consider learning the parameters of a nonlinear uncertainty model of the form in \cref{eq:func-approx} directly using set-membership estimation. Under the prior knowledge \cref{as:prior-conf}, the initial feasible parameter set is given as $\Theta(0) = \{\what(0)\} \oplus \calW(0)$ and the feasible parameter set at time $t$ is obtained as
\begin{equation}\label{eq:est-setmember-fps}
\begin{aligned}
    \Theta(t) = &\big\{W \in \Theta :\,
        x(k+1) - Ax(k) - Bu(k)  \\
        &- W\phi(x(k)) \in \calV,\, \forall k \in \{0, \dots, t-1\}
    \big\}.
\end{aligned}\end{equation}
When $\Theta(0)$ and $\calV$ are a hyperboxes, this estimator maintains independent feasible sets for each row of $W$ and can be updated recursively in time with polytopical set intersections by rewriting \cref{eq:est-setmember-fps} in terms of the row-wise vectorization of $W$. 
Clearly, $\Theta(t) \subseteq \Theta(t-1)$. As is common practice in the literature \cite{bujarbaruah_semi-definite_2020}, we propose generating a point estimate of the parameters as the \emph{Chebyshev center} of the feasible parameter set:
\begin{align}\label{eq:est-chebycenter}
    \what(t) = \argmin_{\what}\max_{W \in \Theta(t)} \|\what - W\|_F.
\end{align}
By definition, this approach minimizes the worst-case error of the point estimates and is typically straightforward to compute \cite{BoydVandenberghe2004}. Denoting the Chebyshev radius for the feasible parameter set associated with the $i$-th row of $W$ as $r_i(t) = \min_{w} \max_{w_i \in \Theta_i(t)} \|w - w_i\|_2$, we take the confidence interval on $\hat w_i(t) - w_i$ as $\calW_i(t) = \{\tilde w_i : \|\tilde w_i\|_2 \leq r_i(t)\}$. 

By definition, since $\Theta(t) \subseteq \Theta(t-1)$, the Chebyshev radii must be decreasing over time: $r_i(t) \leq r_i(t-1)$. Therefore, a set-membership estimator with point estimates as the Chebyshev center satisfies \cref{as:shrinking} with a risk tolerance $\delta = 0$, thus guaranteeing the safety of the system with probability $1$ according to \cref{def:safety}. 

\subsection{Recursive Bayesian Linear Regression (BLR)}\label{sec:blr}
In the case of Bayesian estimation, we can generate confidence intervals directly from the posterior distribution over parameters if we know the disturbance distribution. We outline this approach under a simple, standard assumption.
\begin{assumption}\label{as:est-subgaussian-noise}
We assume that each entry of the process noise is bounded $v(t) = [v_1(t), \dots,v_n(t)]^\top \in \calV = \{v : |v| \leq \sigma_i\}$, and that each entry $v_i(t)$ is independent of the others. Hence, $v_i(t)$ is sub-Gaussian with variance proxy $\sigma_i^2$.  
\end{assumption}

Under \cref{as:est-subgaussian-noise}, we can essentially treat the noise as both normally distributed for convenient analysis and provide safety guarantees for the algorithm proposed in \cref{sec:approach}. If we place subjective priors over the rows of $W$ of the form $w_i \sim \mathcal{N}(\hat{w}_i(0), \sigma_i^2 \Lambda_i^{-1}(0))$, then the resulting posterior parameter distribution at time $t$ is also Gaussian, $w_i \sim \mathcal{N}(\hat{w}_i(t), \sigma_i^2 \Lambda_i^{-1}(t))$. We then use the mean of the posterior---also corresponding to the maximum a posteriori (MAP) estimate---as a point estimate for control: $\hat{f}(x,t) = \what(t)\phi(x)$. Defining the measurement and prediction at time $t$ as $y(t) := x(t+1) - Ax(t) - Bu(t)$ and $\hat{y}(t):= \what(t) \phi(t)$, the MAP estimate for each row can then be updated with constant complexity in time using the recursive updates   
\begin{equation}\label{eq:predict}
\begin{aligned}
    \hat{w}_i(t+1) &= \hat{w}_i(t) - \frac{(\hat{y}_i(t)- y_i(t))\phi(t)^\top \Lambda_i^{-1}(t)}{1 + \phi(t)^\top \Lambda_i^{-1}(t) \phi(t)} \\
    \Lambda_i^{-1}(t+1) &= \Lambda_i^{-1}(t) - \frac{\Lambda_i^{-1}(t)\phi(t) \phi(t)^\top \Lambda_i^{-1}(t)}{1 + \phi(t)^\top \Lambda_i^{-1}(t) \phi(t)}
\end{aligned}~,
\end{equation}
where we write each entry of $\hat{y}(t)$ as $\hat{y}_i(t) = \hat{w}_i(t)^\top \phi(t)$ and $\phi(t) := \phi(x(t))$. 
We can recover the frequentist ordinary least-squares estimator if we assume a flat prior \cite{deisenroth2020mathematics}, which requires the availability of some amount of prior data to yield the initial values $\what(0)$, $\Lambda(0)$. 
        
Taking a risk tolerance of $\delta \in (0,1)$, we could naively define the confidence interval for the $i$-th row of $\what(t)$ as 
\begin{align}\label{eq:est-lstsq-naive}
    \calW^{\mathrm{naive}}_i(t) := \{\tilde w_i \in \R^n : \tilde w_i^\top \Lambda_i(t) \tilde w_i \leq \sigma_i^2 \chi^2_n(1-\frac{\delta}{n})\},
\end{align}
where $\chi_n^2(1-\delta)$ is the $1 - \delta$ quantile of the chi-square distribution with $n$ degrees of freedom. However, the confidence interval in \cref{eq:est-lstsq-naive} does not capture the fact that we want to certify the safety of the policy for all time with high probability. We cannot achieve this with a single confidence interval of a point estimate at time $t$, as \cref{eq:est-lstsq-naive} ignores the correlations between the model estimates over time. As we discussed in our learning desiderata for robust control, we instead desire confidence intervals such that
\begin{equation}\label{eq:est-lstsq-conf}
    \what(t) - W \in \calW(t),\ \forall t \in \mathbb{N}_{\geq 0},
\end{equation}
with probability at least $1 - \delta$.

\rev{Recent work applied a Martingale argument originating from the Bandits literature to generate such confidence intervals by scaling the naive confidence intervals with a time-varying parameter \cite{lew_safe_2020}. The resulting safety guarantees are subject to assumptions on the calibration of the prior, for which we refer the reader to \cite[Assumption 3]{lew_safe_2020}. This assumption---that the prior contains the true parameter with high probability---is trivially satisfied for flat priors (i.e., when we have access to some data a priori). We also assume this assumption is satisfied for the non-flat priors that we use in this work. We may assume that these priors are constructed either via empirical Bayes (as we discuss in \cref{sec:meta}) or via expert knowledge.} 

\begin{theorem}\label{thm:thomas}
\cite[Thm~1]{lew_safe_2020} For the recursive Bayesian linear filter \cref{eq:predict}, we have the estimation error $\hat{w}_i(t) - w_i \in \calW_i(t)$ for all $t\geq 0$ with probability at least $1-\delta$, where
\begin{equation}\label{eq:bls-conf}
    \calW_i(t) 
    := \{\tilde w_i : (\tilde w_i^\top \Lambda_i(t) \tilde w_i)^{\frac{1}{2}} \leq \sigma_i \beta_t(\delta /n)\},
\end{equation}
with $\beta_t(\delta)$ equal to
\begin{equation}
    \sqrt{2 \log \Big(\frac{\det(\Lambda_i(t))^{1/2}}{\delta \det(\Lambda_i(0))^{1/2}}\Big)} + \sqrt{\frac{\lambda_\mathrm{max}(\Lambda_i(0))}{\lambda_\mathrm{min}(\Lambda_i(t))} \chi^2_n(1-\delta)}.
\end{equation}
\end{theorem}

The confidence intervals resulting from \cref{thm:thomas} unfortunately do not immediately satisfy the requirement that $\calW(t+1) \subseteq \calW(t)$ in \cref{as:shrinking} without a \emph{persistence of excitation} or \emph{active exploration} assumption as is made in \cite{mania2020active}. A simple workaround is to update the estimate \cref{eq:predict} fed into to the controller only when the associated confidence intervals \cref{eq:bls-conf} have shrunk, effectively disregarding new data until the system has been excited sufficiently. This approach was shown to perform well in practice \cite{koller_learning-based_2018}. Therefore, we can apply the BLR estimator to guarantee the safety of the system with any desired risk tolerance $\delta \in [0,1]$ according to \cref{def:safety}. Still, future work should explore strategies to guarantee confidence intervals constructed using \cref{thm:thomas} (or other equivalent results) satisfy \cref{as:shrinking} more naturally.  

\subsection{Toy Estimation Example}
We compare the performance of the set-membership estimator \cref{eq:est-chebycenter} with the Bayesian least squares estimator \cref{eq:predict} on the toy problem
\begin{equation}
    y(t) =[w_1, w_2] \begin{bmatrix} \sin(4x_1(t)) \\\tanh(x_2(t))\end{bmatrix} + v(t).
\end{equation}
We set $w = [0.5, 0.5]^\top$, sample $v(t) \iid \mathrm{U}[-0.4, 0.4]$ and generate training samples $x$ with each entry $ x_i(t)\iid \mathrm{U}[-1,1]$. For the set-membership estimator, we set the initial feasible parameter set $\Theta(0) = \{\tilde w : \|\tilde w\|_\infty \leq 1\}$. For the Bayesian least-squares estimator we select a zero-mean prior such that the 95\% confidence interval $\calW(0)=  \{w : \|w\|_2 \leq 1\}$. \cref{fig:estimators-uniform} shows the evolution of the parameter estimates after processing $t = 25$ samples. The set-membership estimator quickly converges to an accurate estimate as measured by the size of the feasible parameter set. In contrast, the BLS estimate is slower to converge with a looser confidence interval.

However, if we add a small bias of $\alpha = 0.05$ to the measurements and run the experiment again, we observe highly undesirable behavior from the set-membership estimator. As shown in \cref{fig:estimators-biased}, the feasible parameter set shrinks to a set that does not include the true parameter. Moreover, in our experiment the feasible parameter set collapsed to the null set after $t=13$ measurements, indicating there were no parameters that could explain the data any longer. In contrast, the BLS estimator performs similar to the well-specified example. These experiments highlight the fragility of the set-membership estimator, which we believe highlights the need for greater consideration of probabilistic estimators such as Bayesian regression in adaptive and learning MPC. Model misspecification in the form of imperfect features and outlier noise samples outside of the disturbance set $\calV$ can generally result in feasible parameter sets that vanish or are erroneous, resulting in the loss of safety guarantees of an adaptive MPC algorithm. \rev{In contrast, our example shows the BLS estimator is more robust to imperfections in the control design, making it a much more practical estimator for real-life applications.}

\begin{figure}[t]
    \centering
        \includegraphics[width=\linewidth]{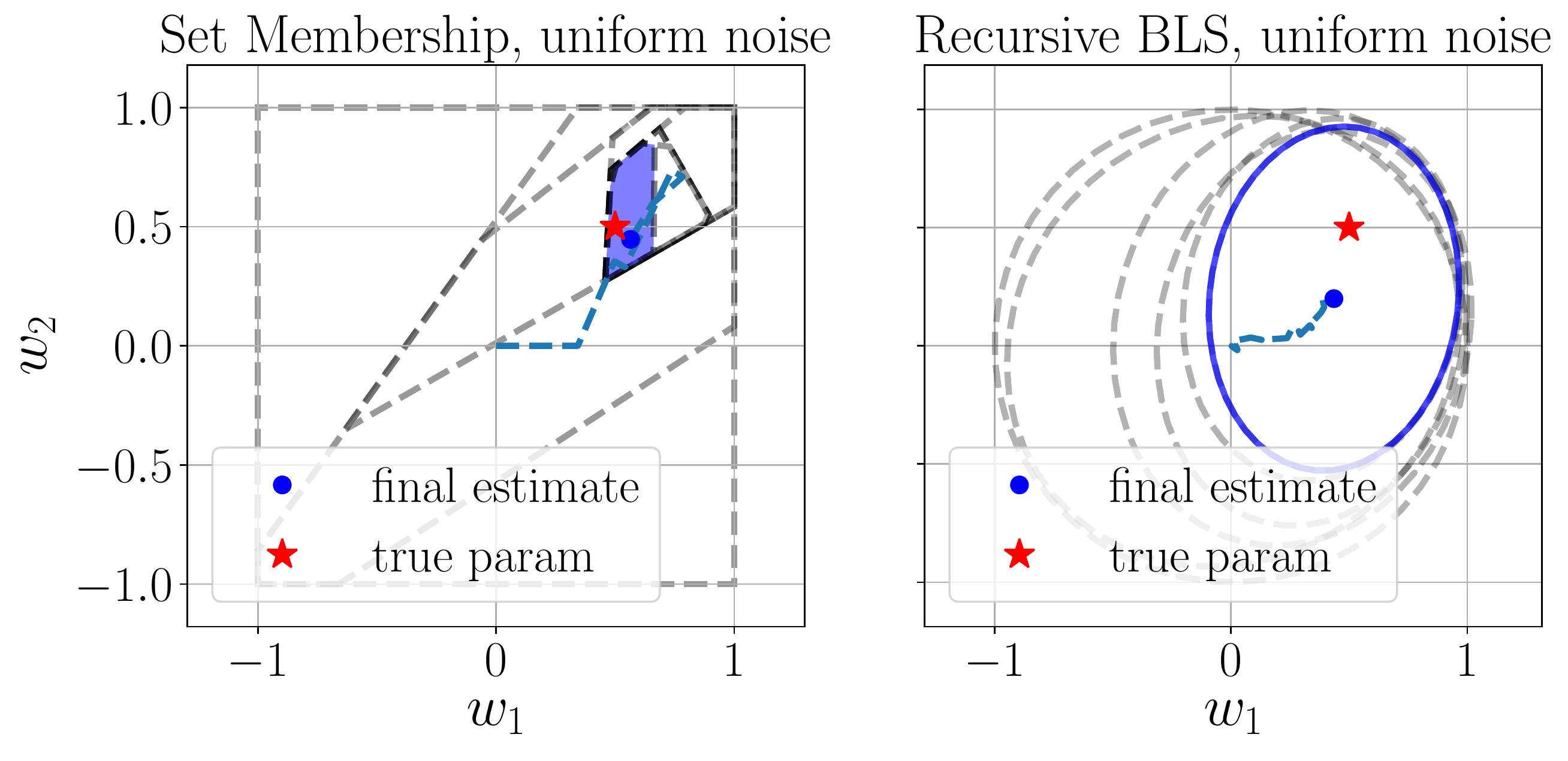}
        \caption{Left: The trajectory of the parameter estimate found using set-membership with feasible parameter sets plotted in gray at 5 sample intervals. Right: The trajectory of the parameter estimate found using BLS with 95\% confidence intervals plotted at 5 sample intervals. }
    \label{fig:estimators-uniform}
\end{figure}

\subsection{Calibrated Priors using Meta-Learning}\label{sec:meta}
\rev{In our problem formulation and throughout \cref{sec:approach}, we considered the uncertain function $f$ in the dynamics \cref{eq:problem-dynamics} as a linear combination of known, fixed basis functions for all time and control iterations. Under this assumption (specifically, \cref{as:bounded-feat}), we designed an adaptive robust MPC that provably guarantees the safety of the system when we use the estimators in \cref{sec:setmembership}-\cref{sec:blr} that satisfy \cref{as:shrinking}. However, $f$ commonly represents the unknown nonlinear effect of an uncertain environment, so a chosen feature representation $\phi$ may not satisfy \cref{as:bounded-feat} exactly in practice. Even if trajectory data in a specific environment is available, a learned feature representation might not capture the nonlinear influence of another environment. }

\begin{definition}
\rev{A \emph{control task} $\mathcal{T} = \{k, x(0), \epsilon\}$ consists of a number of iterations $k \in \mathbb{N}_{\geq0}$, and an initial condition $x(0)$ and parameter $\epsilon \in \R^e$ that represent the \emph{environment}, drawn from a distribution $\epsilon \sim \rho(\epsilon)$. The \emph{environment} $\epsilon$ is constant over a control task and affects the dynamics \cref{eq:problem-dynamics} through the unknown function $f$, so we take the unknown function as $f(x) := f_\epsilon(x)$. }
\end{definition}

\rev{In applications, we typically either do not exactly know how the environment influences the dynamics or we cannot measure $\epsilon$. Thus, a system designer is often left to choose basis functions and encode prior assumptions on the effect size arbitrarily, which may be difficult or inaccurate for complex systems when the controller is deployed in a wide range of operating conditions.}

\rev{Therefore, we propose to cast the problem of learning the structure of the disturbance through the lens of meta-learning algorithms---methods that ``learn to learn''---which have gained prominence in recent years for their ability to learn general representations that can rapidly adapt to new tasks \cite{FinnAbbeelEtAl2017}. These approaches provide an automatic method for identifying features and calibrating priors based on data from the system acting in various environments \cite{harrison_meta-learning_2020}. Moreover, a practitioner could generate such data from simulation, wherein the assumptions on the system dynamics are automatically mapped to numerically convenient basis functions \cite{lew_safe_2020}.}

\begin{figure}[t]
    \centering
        \includegraphics[width=\linewidth]{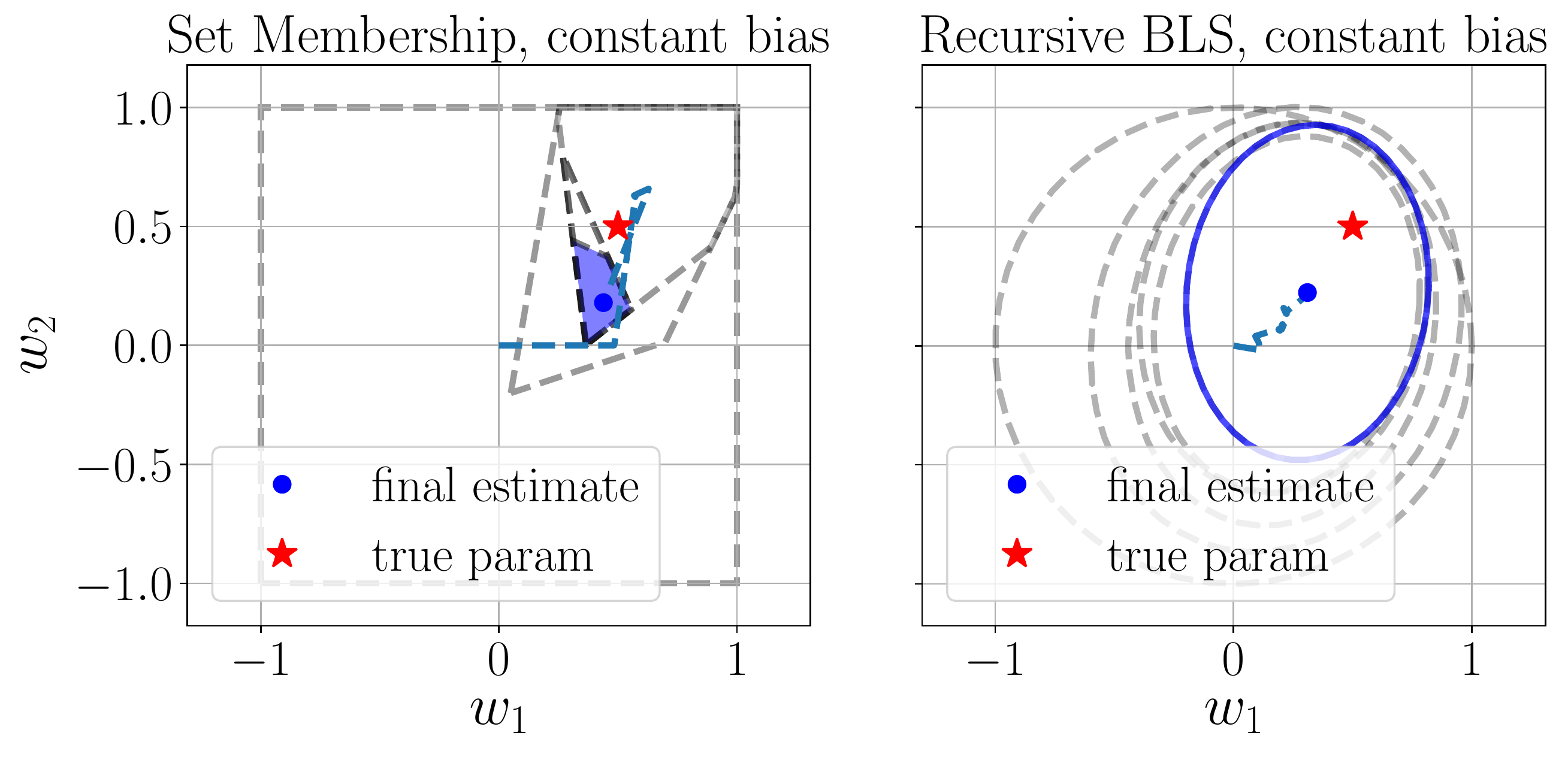}
        \caption{Left: The trajectory of the parameter estimate found using set-membership with feasible parameter sets plotted in gray at 5 sample intervals. Right: The trajectory of the parameter estimate found using BLS with 95\% confidence intervals plotted at 5 sample intervals. }
    \label{fig:estimators-biased}
\end{figure}

\rev{In particular, we propose to leverage the ALPaCA meta-learning algorithm \cite{harrison_meta-learning_2020}. This learning system consists of two components: neural network features, which are fixed for the duration of a control task, and a last layer which is updated online. The model takes the form of a matrix last layer $W \in \R^{n \times d}$, applied to nonlinear neural network features, yielding nonlinear predictive model $\hat{f}(x) = W \phi_\theta(x)$, where $\theta$ denotes the neural network weights. The ALPaCA algorithm learns both the feature representation $\phi_\theta$ and a prior distribution on the last layer $p(W)$ to approximate the distribution over environment uncertainties $f_\epsilon$ induced by $\epsilon \sim \rho(\epsilon)$ as the distribution over models $W\phi_\theta(x)$ induced by the prior $W \sim p(W)$. Following \cite{lew_safe_2020}, we represent the prior with independent Gaussian distributions on each row of $W$.}

\rev{We outline the ALPaCA algorithm's mechanics under the assumption that each control task consists of a single iteration of length $T$ to simplify notation. Let $\mathcal{H}^{(j)} = \{(x^{(j)}(0), u^{(j)}(0)), \dots, (x^{(j)}(T), u^{(j)}(T))\}$ be the trajectory data collected during task $j$. We assume access to a trajectory dataset $\mathcal{H}^{(0)}, \dots, \mathcal{H}^{(k)}$, each associated with a different unknown instantiation $f^{(0)}(\cdot), \dots, f^{(k)}(\cdot)$ of the uncertain function~$f$ in \cref{eq:problem-dynamics} resulting from the changing environments $\epsilon^{(0)}, \dots, \epsilon^{(k)} \iid \rho(\epsilon)$. In the \emph{inner loop} learning algorithm, which is the learning that happens online within one environment, the posterior over $W$ is computed based on an assumption of Gaussian noise, as discussed in \cref{sec:blr}. Then, the \emph{offline outer learning algorithm}  backpropagates through the Bayesian linear regression learning procedure to train the neural network features and learn the prior on $W$. As such, the ALPaCA model learns a set of features that are suited to the problem class expected to occur online, paired with a prior over these features that is calibrated to the environment distribution.}

Feature-based function approximation is a well studied problem in classical adaptive control \cite{SlotineLi1991, IoannouFidan2006}, with methods that generally rely on a small number of hand-crafted basis functions. For applications where the dynamics are poorly understood, such as aircraft wake vortices \cite{lavretsky_workshop}, this makes it necessary to reason about the misspecification of the model even in unconstrained control \cite{IoannouFidan2006, joshi_asynchronous_2020}. Another approach is to use a large number of random features that can approximate any function as the number of features goes to infinity \cite{BoffiTuEtAl2020}. This is undesirable in a constrained control setting, as the confidence intervals we require in \cref{sec:est} scale poorly in the number of features. A meta-learning approach like ALPaCA \cite{harrison_meta-learning_2020} offers a solution to this apparent trade-off, learning an accurate and compact feature representation from data offline. Therefore, we do not consider misspecification of the features in this work, although future work could consider tightening the constraints in the tube MPC \cref{eq:unmatched-mpc} further with an empirical bound on $\|\phi_\theta(x) - \phi(x)\|$  obtained, for example, by a bootstrapping procedure.

\section{Experiments}\label{sec:exp}
In this section, we simulate our approach on several example systems to highlight the various benefits of our methods and illustrate their properties. 

\subsection{Adaptive Robust MPC Benchmark}
\rev{To benchmark our adaptive, robust MPC (ARMPC) algorithm that \emph{matches} as much of the uncertainty as possible, we compare it against the normative approach in ARMPC; we estimate the range of values that the uncertainty $f$ can take and naively treat it as a disturbance using tube MPC as in \cite{bujarbaruah_adaptive_2018, bujarbaruah_semi-definite_2020, soloperto_learning-based_2018, KohlerAndinaEtAl2019}. These methods improve their performance online by refining a non-increasing bound on the range of $f$. To the best of our knowledge, such algorithms have not been proposed for uncertain terms that satisfy \cref{as:bounded-feat} in the literature. Therefore, we apply the tube MPC strategy in \cite{goulart_optimization_2006} by solving \cref{eq:robust-mpc} online using the disturbance set $\calD'(t) = \calFhat(t) \oplus \calV$, since \cref{cor:unmatched-support} implies that both $f(x(t)) + v(t) \in \calD'(t)$ and that $\calD'(t) \subseteq \calD'(t-1)$ for all $t\geq0$. Analogous to our approach, we construct the terminal cost and maximally RPI set $\mathcal{O}(t)$ using $\calD'(t)$ and the fixed (LQR) controller in \cref{as:unmatched-termcost}. It is straightforward to verify the recursive feasibility and input-to-state stability of this adaptive benchmark via a symmetric argument to the proofs of \cref{thm:unmatched-recfeas}, \cref{thm:unmatched-iss}. }

\begin{figure}[t]
    \centering
        \includegraphics[width=\linewidth]{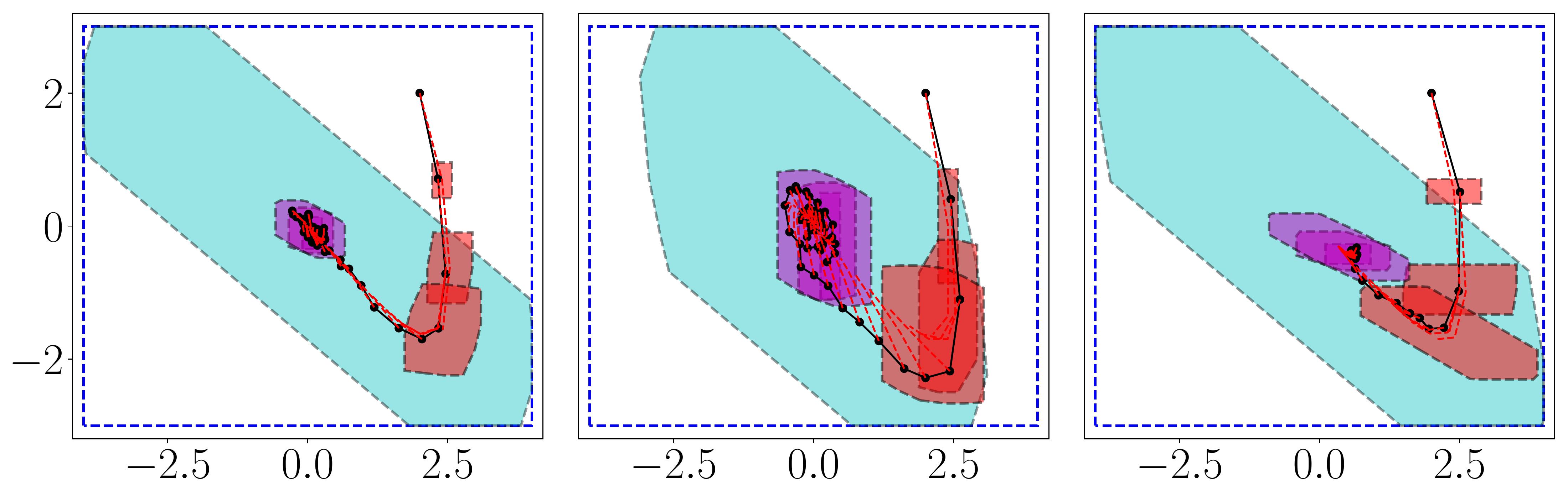}
    \caption{\rev{Phase plots of the closed-loop trajectories of different systems. The realized system evolution is shown in black, the predicted nominal trajectories are depicted in red, as are the predicted reachable sets at $t=0$. The predicted reachable sets at $t=50$ are depicted in purple. The dark blue line indicates the state constraints. The terminal invariant $\mathcal{O}(50)$ is shown in cyan. Left: Adaptive CE MPC on the system with \emph{matched} uncertainty \cref{eq:exp-matched-dyn}. Middle: Benchmark Adaptive MPC on the \emph{matched} system \cref{eq:exp-matched-dyn}. Right:  Adaptive CE MPC on the system with \emph{unmatched} uncertainty \cref{eq:exp-unmatched}.}}  
    \label{fig:trajs}
\end{figure}

\subsection{Double Integrator with Matched Uncertainty}\label{sec:exp-matched}
We illustrate the properties of our algorithm on a double-integrator system
\begin{equation}\label{eq:exp-matched-dyn}
    x(t+1) = \begin{bmatrix}1 & 0.2 \\ 0 & 1 \end{bmatrix}x(t) + \begin{bmatrix}0 \\ 1\end{bmatrix}u(t) + f(x(t)) + v(t),
\end{equation}
\rev{a typical example in the MPC literature that represents simplified second-order dynamics \cite{BorrelliBemporadEtAl2017}.} First, we consider the \emph{matched} uncertainty $f(x) =  [0, w_1]^\top \tanh([0,1]x)$. We estimate the true parameter, $w_1 = 0.5$, online to improve performance. We take the disturbance as an isotropic Gaussian with $\sigma^2 = 5 \times 10^{-3}$ clipped at its 95\% confidence intervals.
In addition, the system is subject to the state and input constraints ${(-4,-3) \preceq x \preceq (4,3)}$ and $-2 \leq u \leq 2$, respectively. We regulate the system to the origin from ${x(0) = [2, 2]^\top}$ while minimizing the quadratic cost function $h(x, u) = x^\top Q x + u^\top Ru$ over a horizon of length $N=3$. We take $Q = I_2$, $R = 1$. \rev{We use the BLR estimator with a flat prior and collect $k = 45$ data points of the system evolution in a unit box near the origin to form an initial estimate of the model parameters and construct $\mathcal{O}(t)$ using Alg. 10.4 in \cite{BorrelliBemporadEtAl2017}.} 

We plot the closed-loop system evolution in \cref{fig:trajs} (left, middle). Our adaptive certainty-equivalent MPC algorithm is able to effectively control the system. In contrast, the benchmark ARMPC can only react to the learned dynamics after $f$ enters the system as a disturbance, resulting in considerably larger closed-loop oscillations and uncertainty on the predicted trajectory. In \cref{fig:trajs}, we also plot the reachable sets associated with the predicted trajectory at the first and last timesteps of the control task. The reachable sets show that our adaptive MPC resolves prediction uncertainty in the system since they shrink over time. \rev{Moreover, we emphasize that online learning does not noticeably improve the performance of the benchmark, as increasing the confidence in the model does not significantly reduce the estimated range of values that the nonlinear function takes.} 

\rev{Furthermore, we see that the shape of the terminal invariant of our adaptive MPC is qualitatively different from that of the benchmark in \cref{fig:trajs} (left, middle). This is because in our method, we tighten the input constraints of the robust input $u^\star$ in the CE law \cref{eq:unmatched-ce-law} to account for the imperfect matching using $\hat{f}$. Our method therefore relinquishes some nominal control authority so that it can account for a smaller disturbance set when we are sufficiently confident in the model.To see this, note that when $\calW = \{0\}$, i.e., when $\what = W$, it holds that} \rev{$\calDhat = (I - BB^\dagger)\calFhat \oplus \calV \subseteq \calFhat \oplus \calV = \calD'$. The simulation shows that matching uncertainty improves the transient convergence behavior of the system.}

In addition, we compare the asymptotic performance of our algorithm with the benchmark ARMPC as a function of~$w_1$, controlling the magnitude of the nonlinearity. To do this, we set the number of data points used to generate the initial model estimate to~$k=10^3$ and collect trajectory rollouts for various values of~$w_1$ from the fixed initial condition. As \cref{fig:cost-envelopes} (left) shows, our ARMPC is feasible for disturbances more than twice the magnitude of those the benchmark ARMPC can tolerate for the given initial condition. In addition, the realized control cost does not differ significantly from the benchmark for values of $w_1$ when both controllers are feasible.  

Finally, we illustrate how our algorithm tolerates larger dynamic uncertainty by comparing the size of the feasible envelope (i.e., the set of initial conditions for which the MPC problem is feasible) as a function of~$w_1$. We set the number of data points to inform our prior to a modest~$k=50$, grid the state-space, and take the feasible region as the convex hull of the initial conditions for which the MPC problem is feasible. Then, we estimate the percentage of states ~$x_0 \in \calX$ in the feasible envelope as the ratio of volumes between the feasible envelope and the state space~$\calX$, illustrated with solid lines in \cref{fig:cost-envelopes} (right). Our CE ARMPC algorithm can tolerate much larger disturbances than the benchmark ARMPC. In these experiments, the feasible envelope of the benchmark becomes empty when the maximal robust invariant is null ($\mathcal{O}(t)=\emptyset$), indicating that there is no subset of~$\calX$ in which the LQR policy associated with the stage cost~$h$ results in provably safe behavior~\cite{BorrelliBemporadEtAl2017}. Hence, \cref{fig:cost-envelopes} highlights the fragility of existing ARMPC approaches under large disturbances.

\begin{figure}[t]
    \centering
        \includegraphics[width=\linewidth]{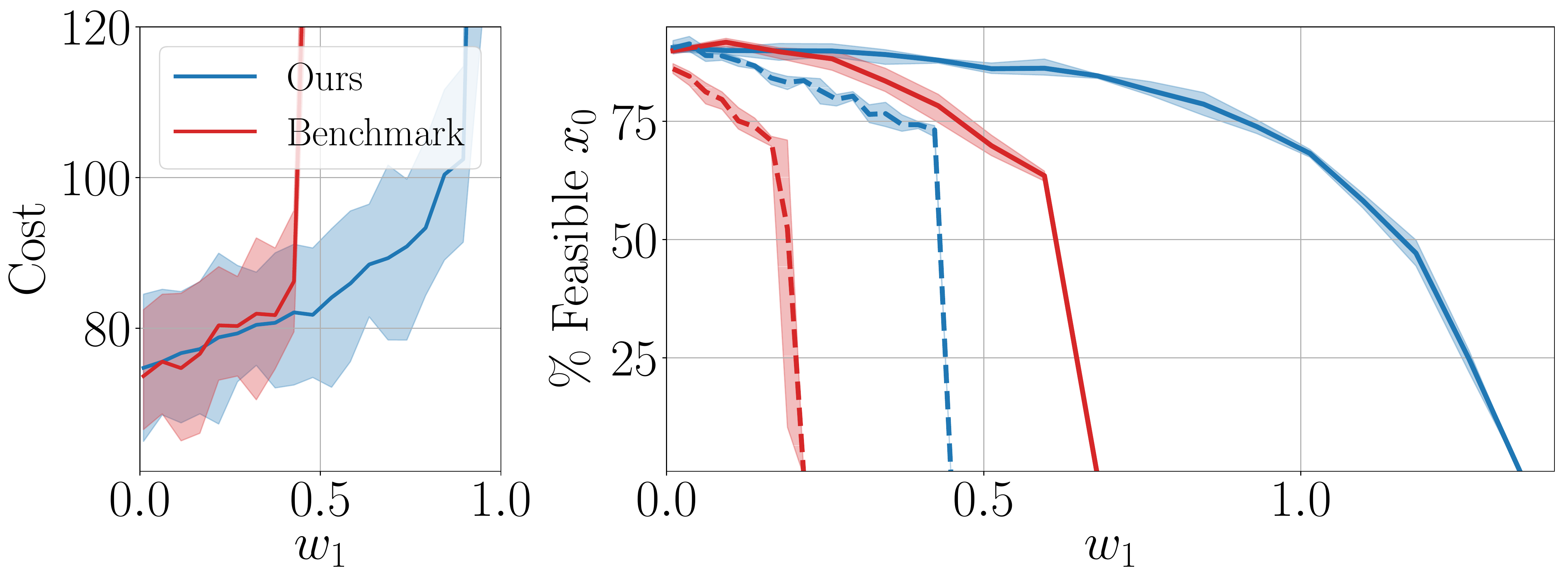}
    \caption{Left: Closed-loop realized trajectory cost for the matched system \cref{eq:exp-matched-dyn} as a function of~$w_1$. Exploding cost indicates infeasibility. The error bars indicate~${2\sigma}$ bounds. Right: Solid lines indicate size of feasible envelopes as a function of $w_1$ for a matched system. Dashed lines indicate the size of the feasible envelopes for the unmatched system \cref{eq:exp-unmatched} as a function of the magnitude of the unmatched dynamics $w_1$ with $w_2=0.5$.}
    \label{fig:cost-envelopes}
\end{figure}

\subsection{Double Integrator with Unmatched Uncertainty}
We now extend the simulations of the simple matched system to the unmatched case to understand the effect of additive nonlinear terms that cannot be canceled from the dynamics. We keep the nominal dynamics identical to \cref{eq:exp-matched-dyn} and take the nonlinear dynamics
\begin{equation}\label{eq:exp-unmatched}
    f(x) = \frac{1}{\sqrt{2}}[w_1 \sin(4x_1), w_2 \tanh(x_2)]^\top,
\end{equation}
where~$w_1$ and~$w_2$ are unknown parameters. Similar to the previous experiments, we initialize the model with~$k=45$ data points sampled around the origin. In this example, the certainty equivalent policy \cref{eq:unmatched-ce-law} can only compensate for the second component of the nonlinear dynamics \cref{eq:exp-unmatched}. We set $w_1 = 0.2$ and $w_2=0.3$. In \cref{fig:trajs} (right), the size of the reachable sets increases if we simulate the system with the unmatched dynamics \cref{eq:exp-unmatched}. The benchmark ARMPC was infeasible from this initial condition, showing that our method still outperforms the benchmark. Next, we fix $w_2=0.5$ and vary $w_1$ to understand the impact of an estimated, unmatched dynamics component. \cref{fig:cost-envelopes} (dashed, right) shows that in our experiment, matching as much of the nonlinear dynamics as possible allows us to handle unmatched dynamics terms of about twice the magnitude as the benchmark. We conclude that our method is a more effective strategy even if the uncertainty is unmatched. Naturally, \cref{fig:cost-envelopes} (right) also shows that the absolute benefit of our method diminishes as the proportion of the nonlinear dynamics~$f(x(t))$ in $\mathrm{Range}(B)$ becomes smaller.

\subsection{Controlling a Planar Quadrotor}
\rev{Next, we simulate a simplified example of a quadrotor in a windy environment. For this example, we examine the episodic setting discussed in \cref{sec:episodic} combined with an application of the meta-learning algorithms in \cref{sec:meta}. We consider a planar version of the quadrotor dynamics for simplicity \cite{Tedrake2021}, with 2D pose $(p_x, p_y, \theta)$ and front and rear thrust inputs $u_f$, $u_r$:}

\rev{
\begin{align*}
    \ddot p_x &= -\frac{1}{m}\sin(\theta)(u_f + u_r) +  \frac{1}{m}f_x(p_x, p_y, \theta)) \\
    \ddot p_y &= \frac{1}{m}\cos(\theta)(u_f - u_r) + \frac{1}{m}f_y(p_x, p_y, \theta)) - g \\
    \ddot \theta &= \frac{l}{I}(u_f - u_r)
\end{align*}
}

\rev{We take the state as $x = [p_x, p_y, \theta, \dot p_x, \dot p_y, \dot \theta]^\top$, linearize the dynamics around~$\bar x=0$,~$\bar u=\frac{mg}{2}[1,1]^\top$, discretize the simulation using Euler's method, and add process noise. The force field $f$ induced by the wind varies spatially, modelling real-world scenarios such as down-wash from another quadrotor. We model the 2D wind disturbance force $f(x) = [f_x(x), f_y(x)]^\top$ as incident at a fixed angle with a maximum velocity $V_w$ that drops of according to an inverse square exponential normal to the direction of incidence $\theta_w$, i.e.,
\begin{equation}
    v_w(x) = V_w\exp\big(-\frac{1}{2}(p_y \cos(\theta_w) - p_x \sin(\theta_w))^2\big),
\end{equation}
resulting in a disturbance force of $\|f(x)\| = c l v_w(x)^2$ along $\theta_w$ for length parameter $l=0.4~\mathrm{m}$ and air resistance $c=0.5~\mathrm{N\cdot s^2 /m^3}$ (we consider the drone's velocity negligible).
We only set constraints on the pose $(p_x,p_y,\theta)$ of the drone. Its linear and angular velocities are unconstrained. The quadrotor is an underactuated system, and therefore the discretized, linearized, simulation has unmatched dynamics terms. Still, similar to the illustration in \cref{fig:drone}, a drone controller can always match disturbance forces along the~$y$ axis in the linearized simulation. }

\begin{figure}[t]
    \centering
    \includegraphics[width=\linewidth]{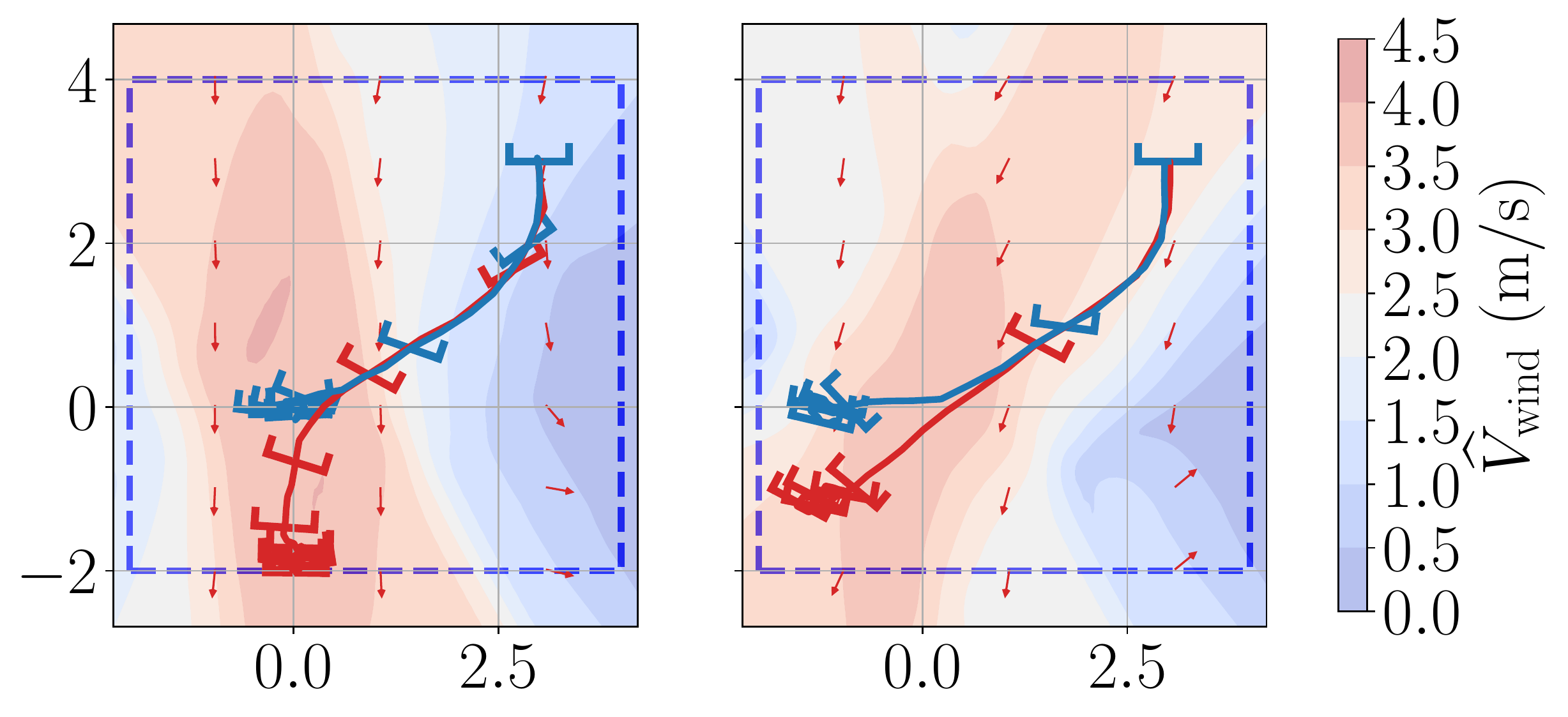}
    \caption{Learned trajectories of our adaptive MPC (blue) compared to a naive tube MPC (red) in the $xy$-plane for a simulated planar quadrotor, with axes in meters. The icons show the orientation of the quadrotor over time. The contours and colorbar indicate the \emph{learned} wind speed in~$\mathrm{m/s}$. Left: Wind comes straight from above. Right: Wind comes at $\theta_w = 22.5^\circ$.}
    \label{fig:drone}
\end{figure}

Spatially varying wind disturbances are challenging to model in practice, so we take a Bayesian approach and model the unknown wind disturbance using a feedforward neural network. To do this, we model the wind disturbance as a linear combination of~$d=5$ sigmoidal output activation functions (scaled by $1/\sqrt{d}$) of a 2 layer network with hidden ReLu activations. 

We generate training data by sampling wind disturbance fields uniformly with wind speeds between $3~\mathrm{m/s}$ and $5~\mathrm{m/s}$ coming in at angles of incidence between $-45^\circ$ and $45^\circ$ relative to the y-axis. We use an ALPaCA model \cite{harrison_meta-learning_2020} to learn an efficient Bayesian representation that can rapidly adapt to any specific conditions encountered in the wild. 



After meta-learning, we control the system for several iterations in the same wind environment using Algorithm \ref{alg:ce-mpc}.C. We compare our adaptive robust MPC with a naive tube MPC similar to the controller in \cref{sec:prelim} that disregards the wind disturbance in the control design, since the support of the estimated wind disturbance was too large for our benchmark ARMPC to be feasible in our experiments. \rev{Therefore, this simulation highlights the fragility of the benchmark approach, which primarily results from non-existence of a terminal invariant set under disturbances of considerable magnitude \cite{BorrelliBemporadEtAl2017}. This phenomenon was also observed previously in applied work \cite{AswaniTomlin2012}, requiring practical workarounds that lose safety guarantees.} 

We run our algorithm with the learned features for $5$ iterations under the same wind conditions. As shown in \cref{fig:drone} (left), if the wind disturbance comes straight from above, the adaptive MPC learns to match the wind forces and reaches the origin quickly. In contrast, the naive tube MPC drifts significantly, as it only reacts to observed disturbances. In addition, if we set the angle of incidence of the wind as~${\theta_w = 22.5^\circ}$, \cref{fig:drone} (right) shows that our approach still achieves decent control performance. \rev{The certainty equivalent controller \cref{eq:unmatched-ce-law} cancels the~$y$-component of the disturbance and converges to a small steady-state offset in the~$x$ direction, showcasing the intuitive behavior of the matching CE controller compared to a naive tube MPC.} In contrast, the benchmark ARMPC algorithm could not guarantee safety for any of the tasks, and a naive unsafe tube MPC that does not consider the wind disturbance at all performs poorly.

\begin{figure}[t]
    \centering
    \includegraphics[width=0.9\linewidth]{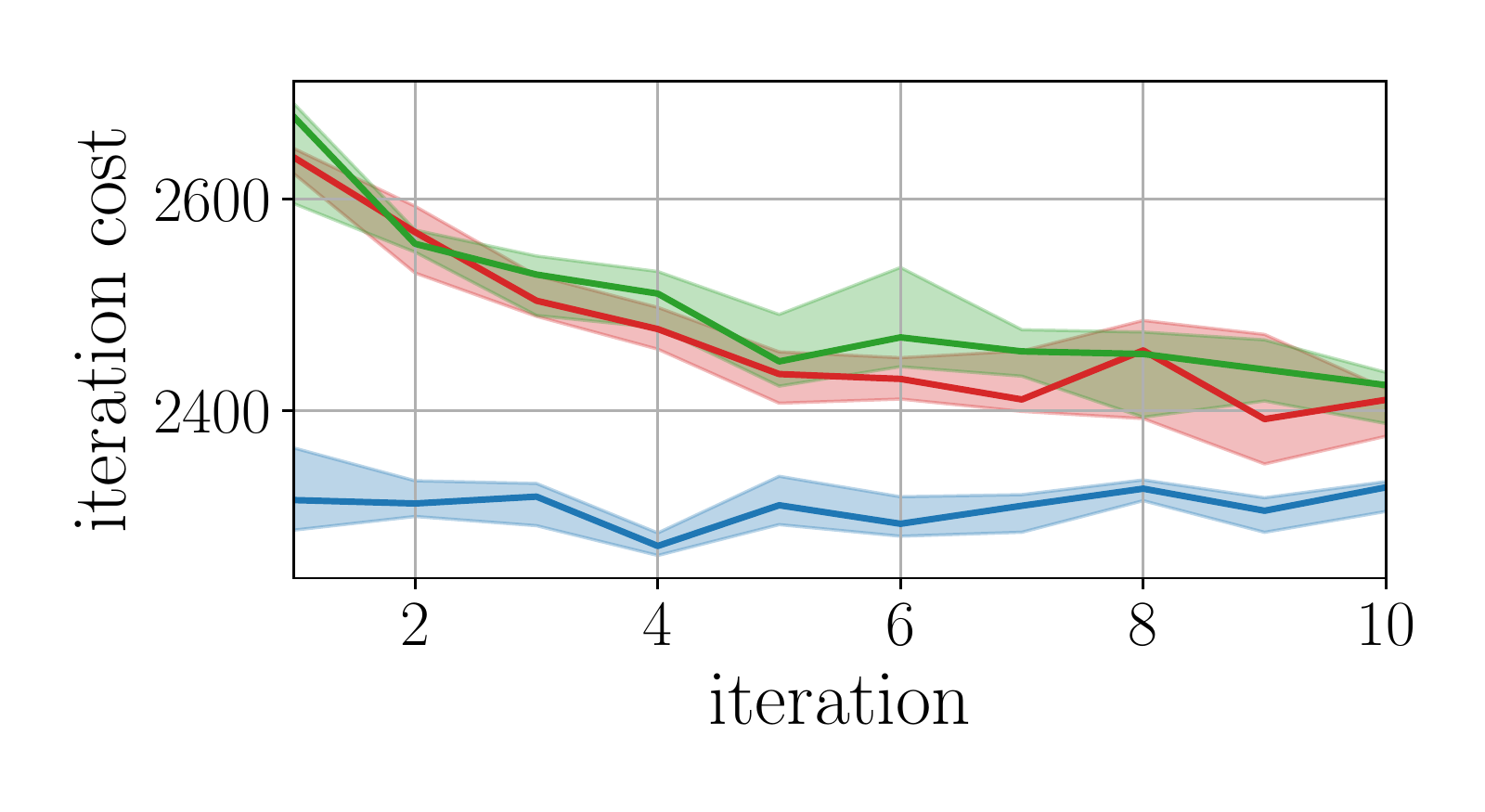}
    \caption{Average closed-loop control cost of the learned adaptive controllers for the simplified drone experiment per control iteration. Error bars indicate $50\%$ quantiles over 10 experiments. Blue: Controller using Bayesian meta-learning. Red: Controller using frequentist meta-learned features. Green: Controller using features that are not meta-learned.}
    \label{fig:meta-curves}
\end{figure}

Finally, we benchmark the efficacy of Bayesian meta-learning algorithms on this problem. We compare the performance of our meta-learned priors against two frequentist counterparts: A frequentist version of ALPaCA that only meta-learns features using a mean-square error (MSE) metric as considered in \cite{OConnellShiEtAl2021}, and a naive model that does not consider the fact that the uncertain function changes between contexts in the training data, directly minimizing the aggregate MSE. To engineer a prior such that our safety guarantees hold for these frequentist models, we fix a zero mean prior with a fixed covariance and pre-train the frequentist models on $k=200$ samples of the test environment---the vertical disturbance in \cref{fig:drone} (left)---such that the bounds on the unknown function match an a priori known bound (as in \cref{rem:proj}, we then also use these bounds in the control design). Our Bayesian model does not receive any additional pre-training data before starting the task. 

As shown in \cref{fig:meta-curves}, the Bayesian meta-learned features and prior result in a model that is able to adapt rapidly to the control task, achieving high performance within a single iteration. This showcases that the ALPaCA meta-learning algorithm can effectively learn a compact feature representation that can quickly learn in new contexts with high confidence. In contrast, the engineered priors on the frequentist models need several trajectory rollouts to drive the cost down.

\subsection{Longitudinal Cruise Control}\label{sec:exp-cruise}
\rev{Finally, we consider the design of a longitudinal cruise controller for an autonomous vehicle (AV). The cruise controller needs to track a fixed reference velocity $v_{\mathrm{ref}}$ and can control the torque $\tau$ applied to the wheels of the vehicle, resulting in the second order dynamics
\begin{equation}
    \ddot p = \frac{1}{m}(R\tau - k \dot p + F_{\mathrm{dist}}(p, \dot p)),
\end{equation}
where $p$ is the position of the car, $m = 10^3 \ \mathrm{kg}$ is its mass, $R = 1/3 \ \mathrm{m}$ is the wheel radius, $k = 20 \ \mathrm{N} \cdot \mathrm{s/m}$ is the friction coefficient, and $F_{\mathrm{dist}}(p, \dot p)$ is an external disturbance force \cite{rajamani2011vehicle}. In this example, we consider the disturbances caused by hilly terrain that the AV passes by on a daily commute, and write the disturbance as induced by an unknown constant incline of $\theta_k$ on the $k$'th road segment $p \in [p_k^1, p_k^2]$ as
\begin{equation}\label{eq:cruise-feat}
\begin{aligned}
    F_{\mathrm{dist}}(p, &\dot p) = \\
    &-\frac{mg}{2}\sum_{k} \sin(\theta_k) \underbrace{\big( \tanh{(p - p_k^1)} + \tanh{(p_{k}^2 - p)}\big)}_{=:2 \sqrt{k}\phi_k(p, \dot p)}
\end{aligned}
\end{equation}
and ignore the air-resistance term presented in \cite{rajamani2011vehicle} for simplicity. We set $\tau = \bar{\tau} - \frac{m}{R} u$, with $u$ as our pseudo-acceleration input into the system around the reference input $\bar\tau = \frac{k}{R}v_{\mathrm{ref}}$, yielding the reference tracking dynamics $x = v_{\mathrm{ref}} - \dot p$ as}
\begin{align*}
    \dot x = - \frac{k}{m} x + u - \frac{1}{m}F_{\mathrm{dist}}(p, v_{\mathrm{ref}} - x).
\end{align*}
\rev{We discretize the dynamics with timestep $dt=.1~\mathrm{s}$ and inject isotropic Gaussian noise with $\sigma^2 = 10^{-3}$ into the position and velocity dynamics. In the spirit of \cref{rem:exogenous}, we treat the position of the vehicle as an \emph{exogenous} signal, regulating and constraining only the velocity and inputs of the vehicle. We collect a $40~\mathrm{s}$ trajectory of training data to initialize a BLR prior over the features defined in \cref{eq:cruise-feat}. In \cref{fig:cruise-control}, we compare the performance of our CE ARMPC with the benchmark, and see that the tracking performance of the controller is improved by learning that the uncertainty in the dynamics is matched.}

\begin{figure}[t]
    \centering
    \includegraphics[width=\linewidth]{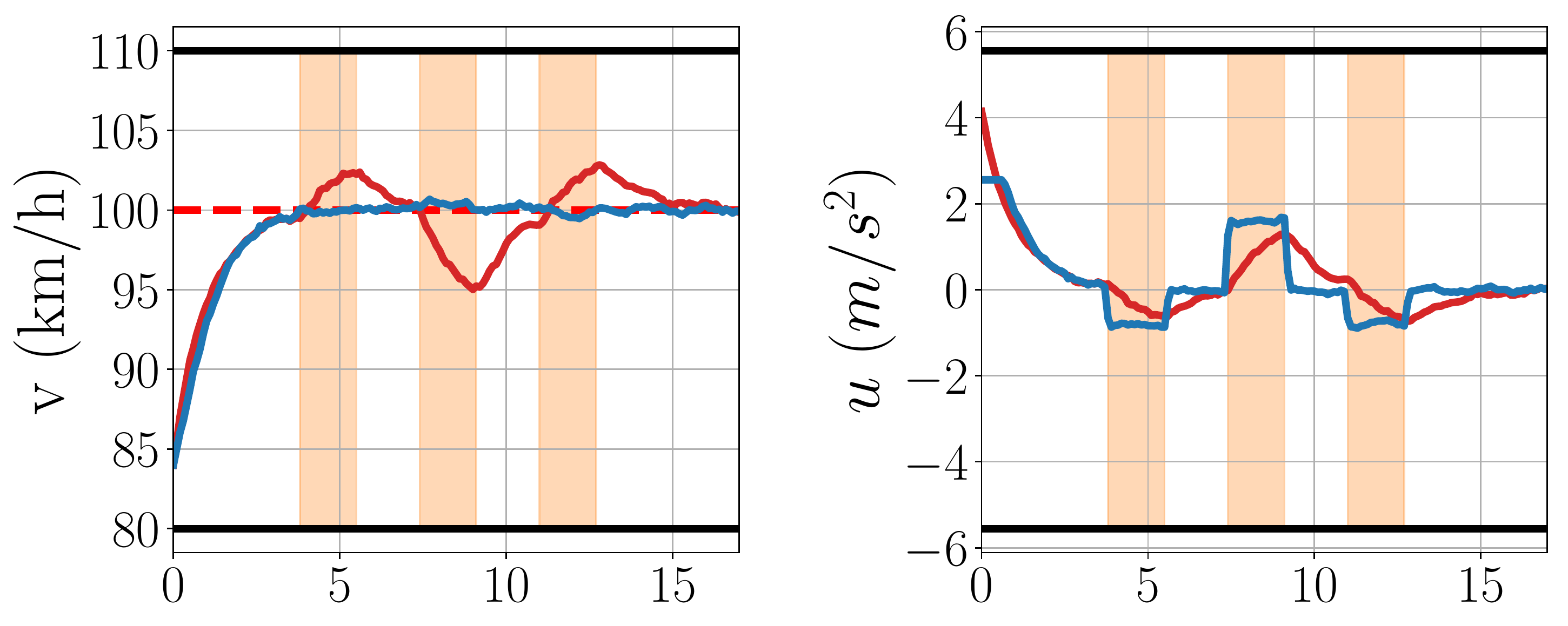}
    \caption{Left: Velocity reference tracking performance over time in seconds. Right: Pseudo-acceleration command over time in seconds. Black: state and input constraints. Red: benchmark ARMPC. Blue: Our CE ARMPC. Red dashed: Reference. Orange: Indicates time spent on inclined road segments.}
    \label{fig:cruise-control}
\end{figure}

\rev{Oscillations in the speed of the vehicle cause passenger discomfort and increase the likelihood of traffic jams. Therefore, we quantify the benefit of our method over the benchmark by comparing the cumulative sum of squared accelerations on the vehicle over the route, a common measure of ride quality similar to ISO standards \cite{svensson2015tuning}. The lower this quantity is, the better the ride quality. \cref{fig:ride-q} shows that our approach improves the ride quality by about $30\%$, a significant improvement. However, \cref{fig:ride-q} also shows that this discrepancy can partially be attributed to the fact that our ARMPC accelerates less aggressively from the initial velocity of $85~\mathrm{km/h}$ in the first $3$ seconds of the simulation. The initial difference in acceleration is caused by the fact that we tightened the MPC's input constraints in \cref{eq:unmatched-mpc} to account for the matching input in the CE policy \cref{eq:unmatched-ce-law}. This is visible in \cref{fig:cruise-control} (right), where the control input applied by our method saturates at a lower constant than the benchmark's inputs in the first second of the simulation. Still, after the initial acceleration, our method clearly outperforms the benchmark on the hilly sections of the route, as \cref{fig:ride-q} shows the relative benefit of our method over the benchmark grows throughout the simulation.} 

\section{Discussion and Conclusions}\label{sec:conc}

The simulations in \cref{sec:exp} show that our method achieves substantial performance improvements compared to existing ARMPC approaches, even when significant components of the nonlinear dynamics are unmatched and cannot be cancelled by a certainty equivalent control policy. We conclude that by extending certainty equivalent control laws from classical adaptive control, we can reduce the conservatism of robust MPC approaches. As a result, we saw in our simulations that our method can tolerate more significant nonlinear terms in the dynamics. In addition, our experiments show that applying the Bayesian meta-learning algorithm ALPaCA \cite{harrison_meta-learning_2020} allows us to learn a feature based representation and a Bayesian prior on the last-layer weights that serve as a sensible engineering solution to satisfy the initial assumptions we make to guarantee the safety of our controller. Moreover, we saw that these models can rapidly adapt to the environmental conditions encountered during deployment.  

Since our control algorithm allows for adaptation laws based on statistical estimation techniques that are more robust to outliers than set-membership estimation, future work should extend our simulations to hardware experiments. In addition, we did not consider a setting in which the uncertain function can change during the control task. Future work can extend our analysis to slowly changing environments, commonly considered using exponential forgetting in adaptive control \cite{SlotineLi1991}. In addition, many systems of interest are inherently nonlinear, therefore, future work should consider applying our method in a constrained nonlinear MPC setting. \rev{Furthermore, as we have assumed in this work, exact measurement of the state is often not possible in applications. Extending learning-based MPC methods for nonlinear systems to the output feedback setting remains a highly complex and largely open problem, which we hope to investigate in future work.}

\begin{figure}[t]
    \centering
    \includegraphics[width=0.9\linewidth]{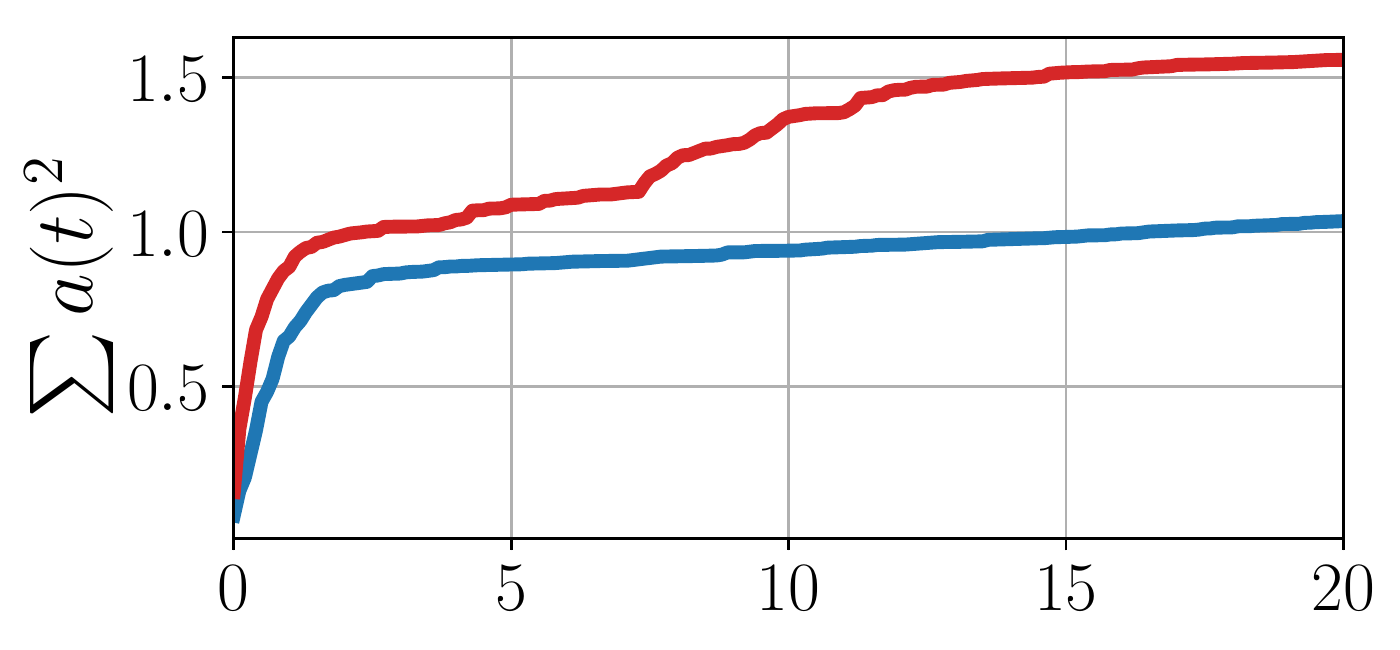}
    \caption{Cumulative sum of squared accelerations experienced by the car in time (seconds), a measure of ride-quality. Blue: Our CE ARMPC. Red: Benchmark ARMPC }
    \label{fig:ride-q}
\end{figure}



\bibliographystyle{unsrt} 
\bibliography{strings,references} 

\newcommand{\noopsort}[1]{} \newcommand{\printfirst}[2]{#1}
  \newcommand{\singleletter}[1]{#1} \newcommand{\switchargs}[2]{#2#1}
\begin{thebibliography}{10}

\bibitem{SinhaHarrisonEtAl2021}
R.~Sinha, J.~Harrison, S.~M. Richards, and M.~Pavone.
\newblock Adaptive robust model predictive control with matched and unmatched
  uncertainty.
\newblock In {\em ACC}, 2022.

\bibitem{rasmussen2006gaussian}
Carl~Edward Rasmussen.
\newblock {\em Gaussian processes in machine learning}.
\newblock MIT press, 2006.

\bibitem{goodfellow2016deep}
Ian Goodfellow, Yoshua Bengio, and Aaron Courville.
\newblock {\em Deep learning}.
\newblock MIT press, 2016.

\bibitem{RichardsAzizanEtAl2021}
S.~M. Richards, N.~Azizan, J.-J.~E. Slotine, and M.~Pavone.
\newblock Adaptive-control-oriented meta-learning for nonlinear systems.
\newblock In {\em RSS}, 2021.

\bibitem{AswaniGonzalezEtAl2013}
A.~Aswani, H.~Gonzalez, S.~S. Sastry, and C.~Tomlin.
\newblock Provably safe and robust learning-based model predictive control.
\newblock {\em {Automatica}}, 2013.

\bibitem{bujarbaruah_adaptive_2018}
M.~Bujarbaruah, X.~Zhang, U.~Rosolia, and F.~Borrelli.
\newblock Adaptive {MPC} for {Iterative} {Tasks}.
\newblock In {\em CDC}, 2018.

\bibitem{bujarbaruah_semi-definite_2020}
M.~Bujarbaruah, S.~H. Nair, and F.~Borrelli.
\newblock A {Semi}-{Definite} {Programming} {Approach} to {Robust} {Adaptive}
  {MPC} under {State} {Dependent} {Uncertainty}.
\newblock In {\em ECC}, 2020.

\bibitem{BujarbaruahZhangEtAl2020}
M.~Bujarbaruah, X.~Zhang, M.~Tanaskovic, and F.~Borrelli.
\newblock Adaptive stochastic {MPC} under time varying uncertainty.
\newblock {\em TAC}, 2021.
\newblock In press.

\bibitem{chowdhary_recursively_2010}
G.~Chowdhary and E.~Johnson.
\newblock Recursively updated least squares based modification term for
  adaptive control.
\newblock In {\em ACC}, 2010.

\bibitem{joshi_asynchronous_2020}
G.~Joshi, J.~Virdi, and G.~Chowdhary.
\newblock Asynchronous {Deep} {Model} {Reference} {Adaptive} {Control}.
\newblock {\em arXiv:2011.02920}, 2020.

\bibitem{HewingKabzanEtAl2017}
L.~Hewing, J.~Kabzan, and M.~N. Zeilinger.
\newblock Cautious model predictive control using {Gaussian} process
  regression.
\newblock {\em TAC}, 2017.

\bibitem{recht_review}
Benjamin Recht.
\newblock A tour of reinforcement learning: The view from continuous control.
\newblock {\em Annual Review of Control, Robotics, and Autonomous Systems},
  2(1):253--279, 2019.

\bibitem{lew_safe_2020}
T.~Lew, A.~Sharma, J.~Harrison, and M.~Pavone.
\newblock Safe {Model}-{Based} {Meta}-{Reinforcement} {Learning}: {A}
  {Sequential} {Exploration}-{Exploitation} {Framework}.
\newblock {\em arXiv:2008.11700}, 2020.

\bibitem{soloperto_learning-based_2018}
R.~Soloperto, M.~A. Müller, S.~Trimpe, and F.~Allgöwer.
\newblock Learning-{Based} {Robust} {Model} {Predictive} {Control} with
  {State}-{Dependent} {Uncertainty}.
\newblock {\em IFAC NMPC}, 2018.

\bibitem{cairano_indirect_2016}
S.~Di Cairano.
\newblock Indirect adaptive model predictive control for linear systems with
  polytopic uncertainty.
\newblock In {\em ACC}, 2016.

\bibitem{DeanManiaEtAl2018}
S.~Dean, H.~Mania, N.~Matni, B.~Recht, and S.~Tu.
\newblock Regret bounds for robust adaptive control of the linear quadratic
  regulator.
\newblock In {\em NeurIPS}, 2018.

\bibitem{DeanTuEtAl2019}
S.~Dean, S.~Tu, N.~Matni, and B.~Recht.
\newblock Safely learning to control the constrained linear quadratic
  regulator.
\newblock In {\em ACC}, 2019.

\bibitem{fan_bayesian_2020}
D.~D. Fan, J.~Nguyen, R.~Thakker, N.~Alatur, A.-a Agha-mohammadi, and E.~A.
  Theodorou.
\newblock Bayesian {Learning}-{Based} {Adaptive} {Control} for {Safety}
  {Critical} {Systems}.
\newblock In {\em ICRA}, 2020.

\bibitem{HarrisonSharmaEtAl2018}
J.~Harrison, A.~Sharma, R.~Calandra, and M~Pavone.
\newblock Control adaptation via meta-learning dynamics.
\newblock In {\em NeurIPS Workshop on Meta-Learning}, 2018.

\bibitem{koller_learning-based_2018}
T.~Koller, F.~Berkenkamp, M.~Turchetta, and A.~Krause.
\newblock Learning-{Based} {Model} {Predictive} {Control} for {Safe}
  {Exploration}.
\newblock In {\em CDC}, 2018.

\bibitem{mania2020active}
H.~Mania, M.~I. Jordan, and B~Recht.
\newblock Active learning for nonlinear system identification with guarantees,
  2020.

\bibitem{MishraGasparino2021}
P.K. Mishra, M.V. Gasparino, A.~Velsasquez, and G.~Chowdhary.
\newblock Deep model predictive control with stability guarantees, 2021.

\bibitem{Astrom2013}
K.~J. {\AA}str{\"o}m and B.~Wittenmark.
\newblock {\em Adaptive control}.
\newblock Courier, 2013.

\bibitem{SlotineLi1991}
J.-J.~E. Slotine and W.~Li.
\newblock {\em Applied Nonlinear Control}.
\newblock {Prentice Hall}, 1991.

\bibitem{IoannouSun2012}
P.~Ioannou and J.~Sun.
\newblock {\em Robust Adaptive Control}.
\newblock {Dover Publications}, 2012.

\bibitem{ChowdharyKingraviEtAl2014}
G.~Chowdhary, H.~A. Kingravi, J.~P. How, and P.~A. Vela.
\newblock Bayesian nonparametric adaptive control using {Gaussian} processes.
\newblock {\em IEEE Transactions on Neural Networks and Learning Systems},
  2014.

\bibitem{BoffiTuEtAl2020}
N.~M. Boffi, S.~Tu, and J.-J.~E. Slotine.
\newblock Regret bounds for adaptive nonlinear control.
\newblock Available at \url{https://arxiv.org/abs/2011.13101}, 2020.

\bibitem{OConnellShiEtAl2021}
M.~O'Connell, G.~Shi, X.~Shi, and S.-J. Chung.
\newblock Meta-learning-based robust adaptive flight control under uncertain
  wind conditions.
\newblock {Available at }\url{https://arxiv.org/abs/2103.01932}, 2021.

\bibitem{BorrelliBemporadEtAl2017}
F.~Borrelli, A.~Bemporad, and M.~Morari.
\newblock {\em Predictive control for linear and hybrid systems}.
\newblock {Cambridge Univ.\ Press}, 2017.

\bibitem{AswaniTomlin2012}
Anil Aswani, Patrick Bouffard, and Claire Tomlin.
\newblock Extensions of learning-based model predictive control for real-time
  application to a quadrotor helicopter.
\newblock In {\em ACC}, 2012.

\bibitem{KohlerAndinaEtAl2019}
J.~K\"ohler, E.~Andina, R.~Soloperto, M.~M\"uller, and F.~Allg\"ower.
\newblock Linear robust adaptive model predictive control: Computational
  complexity and conservatism.
\newblock In {\em CDC}, 2019.

\bibitem{IoannouFidan2006}
P.~Ioannou and B.~Fidan.
\newblock {\em Adaptive Control Tutorial}.
\newblock {SIAM}, 2006.

\bibitem{LavretskyWise2013}
E.~Lavretsky and K.~Wise.
\newblock {\em Robust and Adaptive Control With Aerospace Applications}.
\newblock {Springer}, 2013.

\bibitem{harrison_meta-learning_2020}
J.~Harrison, A.~Sharma, and M.~Pavone.
\newblock Meta-learning {Priors} for {Efficient} {Online} {Bayesian}
  {Regression}.
\newblock In {\em WAFR}, 2018.

\bibitem{mayne:minmax}
P.~O.~M. {Scokaert} and D.~Q. {Mayne}.
\newblock Min-max feedback model predictive control for constrained linear
  systems.
\newblock {\em TAC}, 1998.

\bibitem{mayne_robust_2005}
D.~Q. Mayne, M.~M. Seron, and S.~V. Raković.
\newblock Robust model predictive control of constrained linear systems with
  bounded disturbances.
\newblock {\em Automatica}, 2005.

\bibitem{LimonAlamoEtAl2009}
D.~Limon, T.~Alamo, D.~M. Raimondo, D.~Mu\~{n}oz de~la Pe\~{n}a, J.~M. Bravo,
  A.~Ferramosca, and E.~F. Camacho.
\newblock Input-to-state stability: A unifying framework for robust model
  predictive control.
\newblock In {\em Nonlinear Model Predictive Control: Towards New Challenging
  Applications}, volume 384 of {\em Lecture Notes in Control and Information
  Sciences}, pages 1--26. {Springer}, 2009.

\bibitem{goulart_optimization_2006}
P.~J. Goulart, E.~C. Kerrigan, and J.~M. Maciejowski.
\newblock Optimization over state feedback policies for robust control with
  constraints.
\newblock {\em Automatica}, 2006.

\bibitem{li_input--state_2018}
H.~Li, A.~Liu, and L.~Zhang.
\newblock Input-to-state stability of time-varying nonlinear discrete-time
  systems via indefinite difference {Lyapunov} functions.
\newblock {\em ISA Transactions}, 2018.

\bibitem{jiang_input--state_2001}
Z.-P. Jiang and Y.~Wang.
\newblock Input-to-state stability for discrete-time nonlinear systems.
\newblock {\em Automatica}, 2001.

\bibitem{balas_lpv}
Gary Balas.
\newblock Linear, parameter-varying control and its application to aerospace
  systems.
\newblock {\em ICAS congress proceedings}, 2002.

\bibitem{milanese1991optimal}
Mario Milanese and Antonio Vicino.
\newblock Optimal estimation theory for dynamic systems with set membership
  uncertainty: An overview.
\newblock {\em Automatica}, 1991.

\bibitem{milanese2013bounding}
Mario Milanese, John Norton, H{\'e}l{\`e}ne Piet-Lahanier, and {\'E}ric Walter.
\newblock {\em Bounding approaches to system identification}.
\newblock Springer, 2013.

\bibitem{BoydVandenberghe2004}
S.~Boyd and L.~Vandenberghe.
\newblock {\em Convex Optimization}.
\newblock {Cambridge Univ.\ Press}, 2004.

\bibitem{deisenroth2020mathematics}
M.~R. Deisenroth, A.~A. Faisal, and C.~S. Ong.
\newblock {\em Mathematics for machine learning}.
\newblock Cambridge University Press, 2020.

\bibitem{FinnAbbeelEtAl2017}
C.~Finn, P.~Abbeel, and S.~Levine.
\newblock Model-agnostic meta-learning for fast adaptation of deep networks.
\newblock In {\em ICML}, 2017.

\bibitem{lavretsky_workshop}
E.~Lavretsky.
\newblock Adaptive control: Introduction, overview, and applications.
\newblock {\em ACC Robust and Adaptive Control Workshop}, 2008.

\bibitem{Tedrake2021}
R.~Tedrake.
\newblock Underactuated robotics: Algorithms for walking, running, swimming,
  flying, and manipulation.
\newblock {Available at }\url{http://underactuated.mit.edu}, 2021.

\bibitem{rajamani2011vehicle}
Rajesh Rajamani.
\newblock {\em Vehicle dynamics and control}.
\newblock Springer Science \& Business Media, 2011.

\bibitem{svensson2015tuning}
Lars Svensson and Jenny Eriksson.
\newblock Tuning for ride quality in autonomous vehicle: Application to linear
  quadratic path planning algorithm, 2015.

\end{thebibliography}

\begin{IEEEbiography}[{\includegraphics[width=1in,height=1.25in,clip,keepaspectratio]{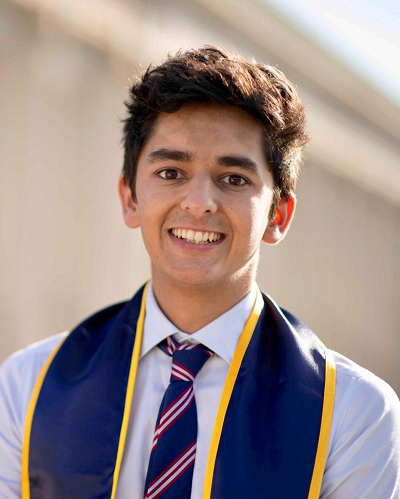}}]{Rohan Sinha} is a Ph.D.~candidate in the Autonomous Systems Lab at Stanford University. He received a B.S. in Mechanical Engineering and B.A. in Computer Science in 2020, both from the University of California, Berkeley. Rohan's research interests lie at the intersection of control theory, machine learning, and applied robotics. Currently, his research focuses on developing learning-based control algorithms with safety guarantees.
\end{IEEEbiography}
\begin{IEEEbiography}[{\includegraphics[width=1in,height=1.25in,clip,keepaspectratio]{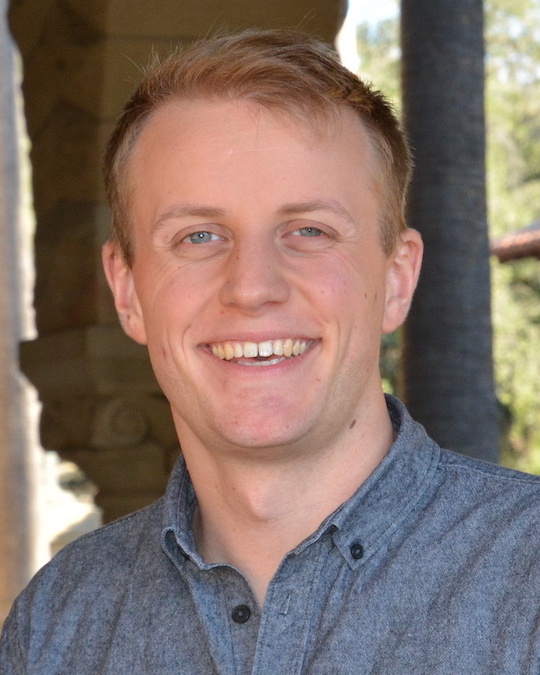}}]{James Harrison}
is a Ph.D.~candidate in the Autonomous Systems Lab at Stanford University. He received an M.S.~degree from Stanford University in 2018 and a B.Eng.~degree from McGill University in 2015, both in mechanical engineering. His research interests include few-shot, adaptive, and open-world learning, Bayesian deep learning, and applications in safe robot autonomy, decision-making, and control. 
\end{IEEEbiography}
\begin{IEEEbiography}[{\includegraphics[width=1in,height=1.25in,clip,keepaspectratio]{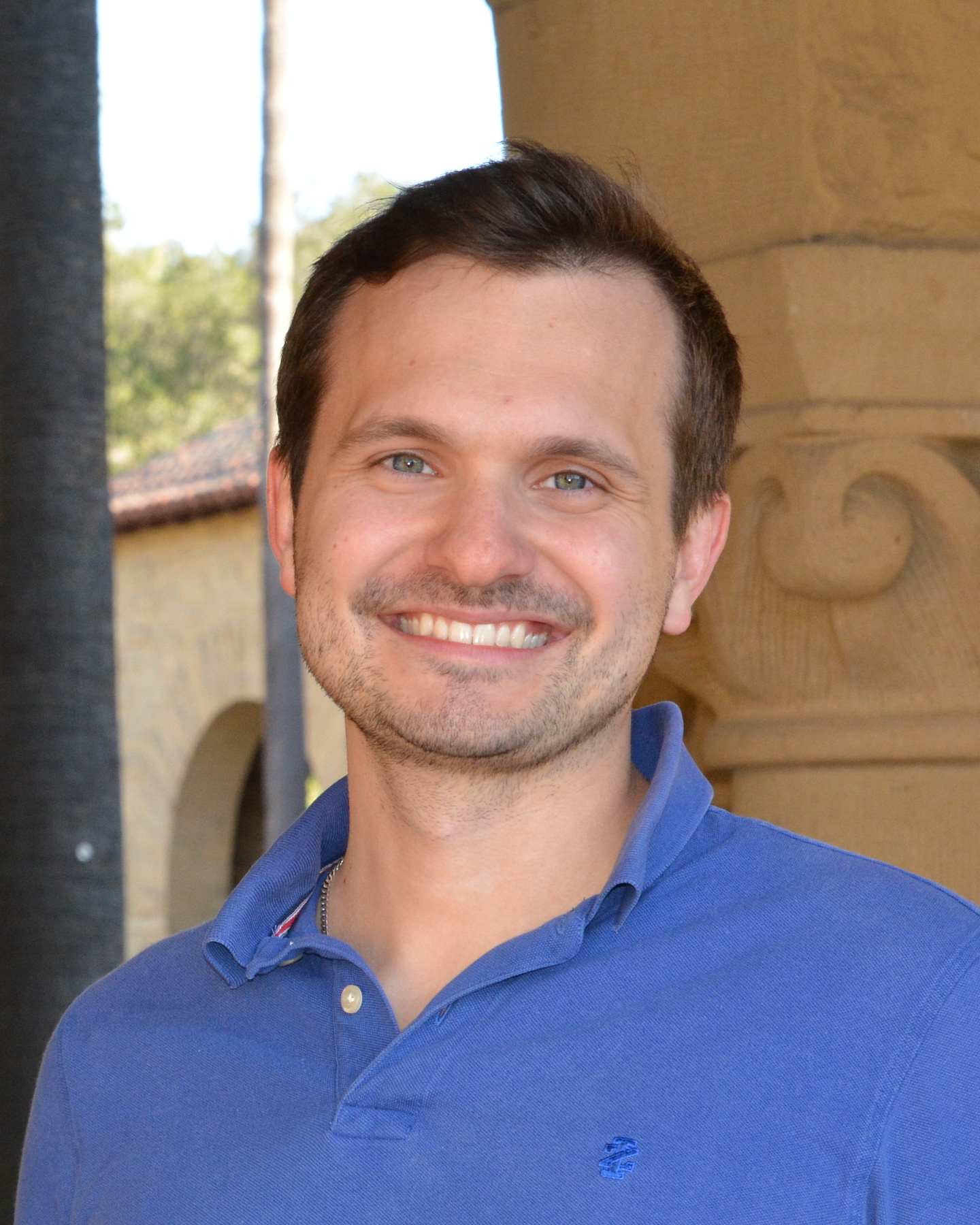}}]{Spencer M. Richards}
is a Ph.D.~candidate in the Autonomous Systems Lab at Stanford University. He received an M.Sc.~degree in Robotics, Systems, and Control from ETH Z\"{u}rich in 2018, and a B.A.Sc.~degree in Engineering Science from the University of Toronto in 2016. His research interests lie at the intersection of control theory and machine learning for robotics, where he works to blend ideas such as adaptive control with meta learning, and stability theory with reinforcement learning.
\end{IEEEbiography}
\begin{IEEEbiography}[{\includegraphics[width=1in,height=1.25in,clip,keepaspectratio]{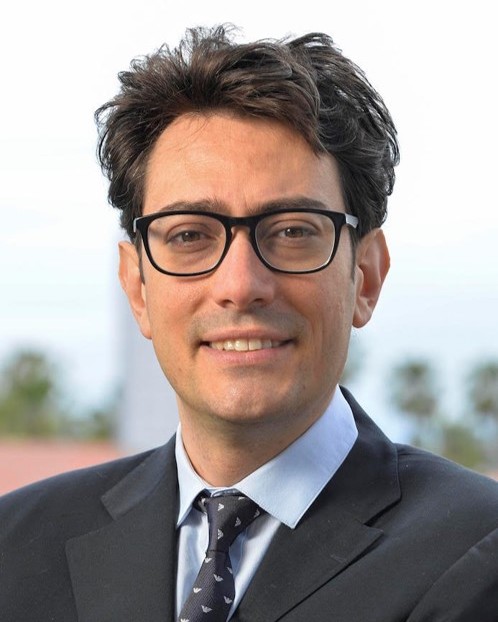}}]{Marco Pavone} 
is   an   Associate Professor of Aeronautics and Astronautics at Stanford University,  where he is the Director of the Autonomous Systems Laboratory. Before  joining  Stanford, he  was  a  Research  Technologist  within the  Robotics  Section  at  the  NASA  Jet Propulsion  Laboratory.   He  received  a Ph.D. degree in Aeronautics and Astronautics from the Massachusetts Institute of  Technology  in  2010.   His  main  research  interests  are  in the  development  of  methodologies  for  the  analysis,  design, and  control  of  autonomous  systems,  with  an  emphasis  on self-driving cars, autonomous aerospace vehicles, and future mobility systems.  
He is a recipient of a number of awards, including  a  Presidential  Early  Career  Award  for  Scientists and Engineers, an ONR YIP Award, an NSF CAREER Award, and a NASA Early Career Faculty Award.  He was identified by the American Society for Engineering Education (ASEE) as one of America’s 20 most highly promising investigators under the age of 40.  He is currently serving as an Associate Editor for the IEEE Control Systems Magazine.
\end{IEEEbiography}

\end{document}